\documentclass[12pt,english, onecolumn, draftcls]{IEEEtran}%
\usepackage[fleqn]{amsmath}
\usepackage{graphicx}
\usepackage{amssymb}
\usepackage{amsfonts}
\usepackage{amscd}
\usepackage{multirow}
\usepackage{color}
\usepackage{babel}
\usepackage{subfigure}
\usepackage{algorithm}
\usepackage{algorithmicx}
\usepackage{listings}
\usepackage{epstopdf}
\usepackage{epsfig}
\usepackage{algpseudocode}
\usepackage{array}
\usepackage{psfrag}
\usepackage{color}

\newtheorem{lemma}{Lemma}

\begin{document}
\IEEEoverridecommandlockouts

\title{Energy Harvesting Two-Hop Communication Networks}
\author{\IEEEauthorblockN{Oner Orhan and Elza Erkip\\}
\IEEEauthorblockA{NYU Polytechnic School of Engineering, Brooklyn, NY\\
Email: \text{onerorhan@nyu.edu}, \text{elza@nyu.edu}}}
\maketitle

\begin{abstract}
Energy harvesting multi-hop networks allow for perpetual operation of low cost, limited range wireless devices. Compared with their battery operated counterparts, the coupling of energy and data causality constraints with half duplex relay operation makes it challenging to operate such networks. In this paper, a throughput maximization problem for energy harvesting two-hop networks with decode-and-forward half-duplex relays is investigated. For a system with two parallel relays, various combinations of the following four transmission modes are considered: Broadcast from the source, multi-access from the relays, and successive relaying phases I and II. Optimal transmission policies for one and two parallel relays are studied under the assumption of non-causal knowledge of energy arrivals and finite size relay data buffers. The problem is formulated using a convex optimization framework, which allows for efficient numerical solutions and helps identify important properties of optimal policies. Numerical results are presented to provide throughput comparisons and to investigate the impact of multiple relays, size of relay data buffers,  transmission modes, and energy harvesting on the throughput.
\end{abstract}

\section{Introduction}\label{intro}
Energy harvesting presents a new paradigm for continuous operation of communication systems without the need for  battery replacement. Energy harvesting technology reduces the operational cost and allows off-grid deployment of sensor nodes such as the ones used within a human body, in nature, or on various structures. As a result, wireless nodes with energy harvesting capability are able to provide long-term data acquisition and monitoring of biological signals, environment and wildlife. An important issue in realizing energy harvesting networks is  the  stochastic nature of energy arrivals with low energy amounts. Therefore, the main concern in energy harvesting wireless sensor network design is the efficient use and management of the harvested energy.

Energy harvesting wireless sensor networks are typically operated over multiple hops to provide range extension and to lower power consumption which favors multiple short hops as opposed to one long hop. Operation over multiple hops brings in another challenge for efficient use of  harvested energies; now multiple nodes have to be coordinated to allow for energy and data causality over each hop,  necessitating the half-duplex relays to switch from reception to transmission modes as a function of the energy and buffer state of the whole network.  The main goal of this paper is to study this problem in the case of two-hop networks involving one or two parallel relays under the {\em offline} optimization framework, which allows for non-causal knowledge of energy arrivals at all the nodes; see \cite{jsac} and references therein for a detailed overview of offline energy harvesting communications systems. While assuming non-causal knowledge presents a simplified model, it allows us to uncover some of the important properties of optimal transmission policies, which determine when and how to use the relays optimally. The insights gained from our work can be used to move towards more practical solutions involving more hops and non-causal knowledge of energy arrivals as done in \cite{elza}.

\subsection{Contributions}
In this paper we investigate two-hop energy harvesting networks with half-duplex relay nodes that have limited size data buffers. We assume the relays employ {\em decode-and-forward} strategy, which is easy to implement in practice. Our goal is to maximize the total throughput delivered to the destination by a deadline. We first study the single relay case as shown in Figure \ref{system:subfig1}. Under the offline optimization framework, we formulate a convex optimization problem and using the Karush-Kuhn-Tucker (KKT) conditions provide properties of optimal transmission policy that determines source and relay schedules and energy levels.

We next consider a two-hop network with two parallel relays \cite{schein}, also known as the {\em diamond relay channel} as shown in Figure \ref{system:subfig2}. The capacity of the diamond relay channel is not known, and the highest achievable rates are based on various combinations of the following four transmission modes \cite{Sandhu}, \cite{Khandani}: i) Broadcast mode, in which the source ($S$) transmits and relays ($R_1$ and $R_2$) listen; ii) the multi-access mode,  in which $R_1$ and $R_2$ transmit and the destination ($D$) listens; iii) successive relaying phase I, in which $S$ and $R_2$ transmit, and $R_1$ and $D$ listen; iv) successive relaying phase II, in which $S$ and $R_1$ transmit, and $R_2$ and $D$ listen. We formulate a convex optimization problem that considers all  four transmission modes jointly. In order to get insights, we investigate some important special cases: i) Successive relaying  phases I and II, also known as \emph{multihop with spatial reuse}; ii) \emph{broadcast and multihop with spatial reuse}; iii) \emph{multi-access and multihop with spatial reuse}. Using the convex optimization framework, we show that optimal transmission policies for the parallel relay case exhibit some characteristics that are different their single relay counterparts. Finally, solving the optimization problems, we illustrate the effect of multiple relays and energy harvesting on the throughput. We also study the impact of the relay data buffer size on performance.

\subsection{Related Work}
In recent years, there has been a surge of interest in energy harvesting communication systems where a significant effort has been devoted to the offline optimization framework;  see \cite{jsac} and \cite{deniz3} for a review of the recent developments. Here we summarize the papers that are closely related to our work. Optimal transmission policies for energy harvesting two-hop networks have been studied in \cite{kaya3}-\cite{yaming}. In \cite{kaya3} two-way relay channels with energy harvesting nodes are considered. Gunduz and Devillers study offline throughput maximization for two-hop communication with a full-duplex relay and with a half-duplex relay for single energy arrival at the source and multiple energy arrivals at the relay in \cite{deniz}. Similarly, multiple energy arrivals at the source
and single energy arrival at the half-duplex relay is studied in \cite{Letaief}. Our previous works \cite{oner}-\cite{oner2} also focus on a half-duplex relay, and for two energy arrivals at the source and multiple energy arrivals at the relay, identify necessary properties of an optimal transmission policy using heuristic arguments. In \cite{oner3} we extend our work in \cite{oner}-\cite{oner2} to include a convex optimization formulation for the case of a single relay and two relays employing multi-hop with spatial reuse. We also provide properties of optimal transmission policies using KKT conditions. The impact of data buffer size for a battery operated relay and a relay with one energy arrival is studied in \cite{varan}. In addition, the throughput maximization problem with amplify and forward
relaying, and relay selection problem are studied in \cite{minasian} and \cite{yaming}, respectively, with non-causal and causal channel and energy arrival information. In \cite{Huang} Huang et. al. study the throughput maximization problem for  the energy harvesting Gaussian relay
channel and Yuyi et. al. in \cite{yuyi} investigate link-selection problem to minimize the average outage
probability. Gurakan et. al \cite{gurakan} consider energy harvesting multi-hop communication
with energy cooperation, where the source can transfer some of its harvested energy to the relay.
 Along this line of work, the throughput maximization problem for
two-hop energy harvesting network with energy transfer from the source to the relay, and
with two-way energy transfer from multiple source nodes are investigated in \cite{berk} and \cite{kaya2},
respectively.

Compared with our conference publications \cite{oner}-\cite{oner3}, this paper introduces a more comprehensive framework to study the parallel relay case by introducing all four transmission modes and by providing a detailed analysis of the optimal transmission policies. We also incorporate the data buffer size limitation at the relays. Furthermore,  the numerical results are extended to include comparisons of various combination of the transmission modes, and impact of number of relays and relay data buffer size on performance.

\begin{figure}
\centering
\subfigure[]{
\includegraphics[scale=0.6,trim= 0 -57 -30 0]{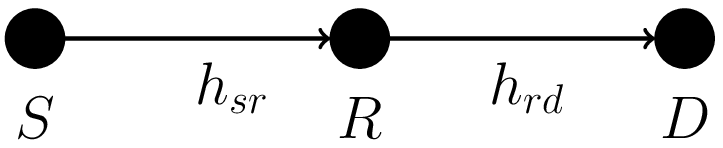}
\label{system:subfig1}
}
\subfigure[]{
\includegraphics[scale=0.6,trim= -30 0 0 0]{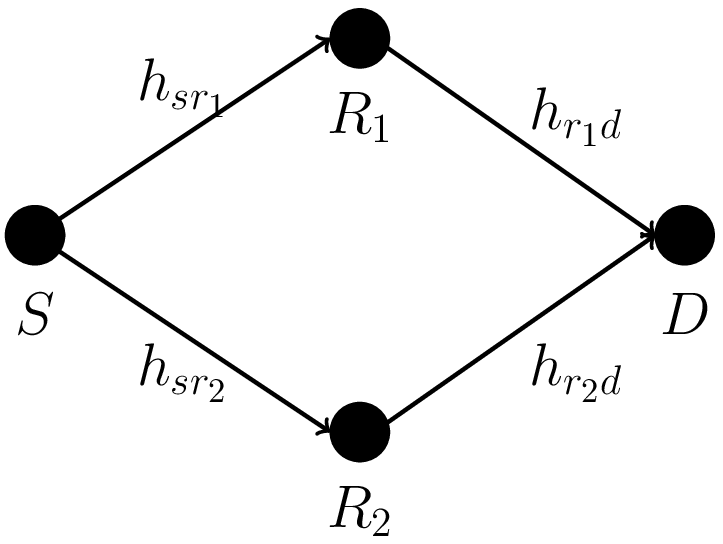}
\label{system:subfig2}
}
\caption{ Two-hop communication with (a) one relay, (b) two parallel relays (diamond relay channel). The power gain between nodes $k$ and $l$ is $\alpha_{kl}=|h_{kl}|^2$ where $h_{kl}$ is the complex channel gain, $k=s,r_1,r_2,r$ and $l=r_1,r_2,r,d$.}
\label{system fig}%
\end{figure}

\subsection{Organization of the Paper}
The paper is organized as follows. In the next section, we describe the system model and
achievable rates for one relay and two relay cases together with some general properties of optimal
transmission policies. In Section \ref{two hop}, we provide a convex formulation and investigate optimal
transmission policies for throughput maximization in the one relay case. We also provide some properties of optimal power
allocation. We formulate a convex problem for the case of two relays in Section \ref{two relay}. We
investigate optimal transmission policies for multi-hop with spatial reuse, broadcast and multihop
with spatial reuse, and multi-access and multi-hop with spatial reuse in Sections \ref{successive relaying}, \ref{BC and Successive relaying mode}, \ref{MAC and Successive relaying mode}, respectively. In Section \ref{numerical}, numerical results are presented, and Section \ref{conclude} concludes the paper.

\section{Preliminaries}
\subsection{System Model}\label{system}
We consider two-hop communication with energy harvesting source ($S$), and one ($R$) or two parallel ($R_1$ and $R_2$) energy harvesting half-duplex relays as in Figure \ref{system fig}. We assume that the relays have finite size data buffer with capacity $B_{max}$ bits. There is no direct link between the source and the destination, and the relays cannot hear one another as in \cite{Khandani}. Each link is modeled as having independent additive white Gaussian noise with unit variance. The complex channel gain between node $k$ and $l$ is $h_{kl}$ where $k=s,r_1,r_2, r$, and $l=r_1, r_2, r, d$, and remains constant throughout transmission. The corresponding power gains are $\alpha_{kl}= |h_{kl}|^2$. For the two relay case, without loss of generality we assume $\alpha_{sr_1} > \alpha_{sr_2}$. We assume that energy arrives at the source and relays with arbitrary and finite amounts at arbitrary times until a given deadline $T$ seconds. For ease of exposition, we combine all energy arrivals at the nodes in a single time series $t_0=0, \ldots, t_K < T$ by allowing zero energy arrivals at some time instants at which only one of the nodes harvests energy.
We denote harvested energy amounts at time $t_i$ by $E_{s,i}$, $E_{r_1,i}$, and $E_{r_2,i}$ for $S$, $R_1$ and $R_2$, respectively, ($E_{r,i}$ for one relay), $i=1,...,K$. In addition, we assume that each node has separate infinite size battery and harvested energies are stored in the batteries without any energy loss. We also assume that there is no energy loss in retrieving energy from the batteries. The time interval between two consecutive energy arrivals $t_{i-1}$ and $t_i$ is denoted by $\tau_i\triangleq t_i-t_{i-1}$, and it is called the $i$'th \textit{epoch.}

Our goal is to maximize the total data delivered to the destination by a given deadline $t=T$ which is referred to as \emph{the throughput maximization} problem \cite{deniz3}. We consider  \textit{offline optimal transmission policies}, that is, we identify optimal power allocation for each node and the transmission schedule assuming that all energy amounts and arrival times are known at the nodes before transmission starts. Here, the transmission schedule indicates which node transmits when, and it is necessary to coordinate the operation of the half-duplex relays. We assume that the nodes consume energy only for transmission.  Due to energy arrivals over time, any feasible transmission policy must satisfy \textit{energy causality} constraints. Energy causality constraints refers to the restriction on the total consumed energy of a node at time $t$ which should be less than or equal to the total harvested energy at that node by that time. In addition, there are \textit{data causality} and \textit{finite data buffer} constraints on the feasible transmission policy. The data causality constraint states that data transmitted by any of the relays up to time $t$ should not exceed total data received by that relay up to that time. The finite data buffer constraint suggests that each relay can store at most $B_{max}$ bits of data in its buffer. We assume both relays have the same size data buffer for simplicity; our results can easily be extended to the case when each relay buffer is of different size.

\begin{figure}
\centering
\subfigure[]{
\includegraphics[scale=0.6,trim= 0 0 0 0]{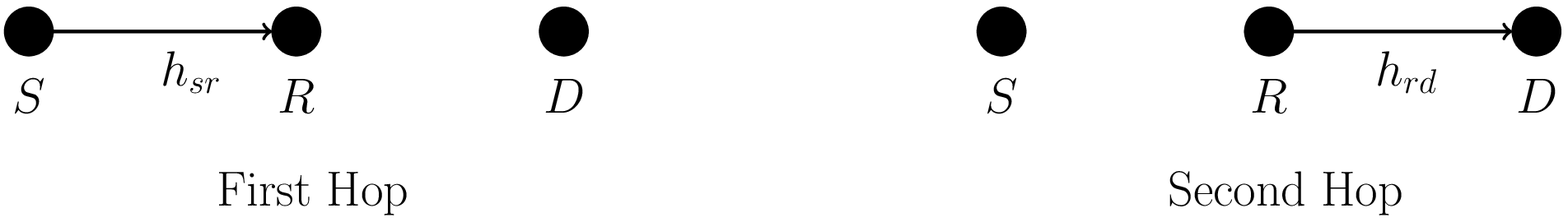}
\label{sys:subfig1}
}
\subfigure[]{
\includegraphics[scale=0.6,trim= 0 0 0 0]{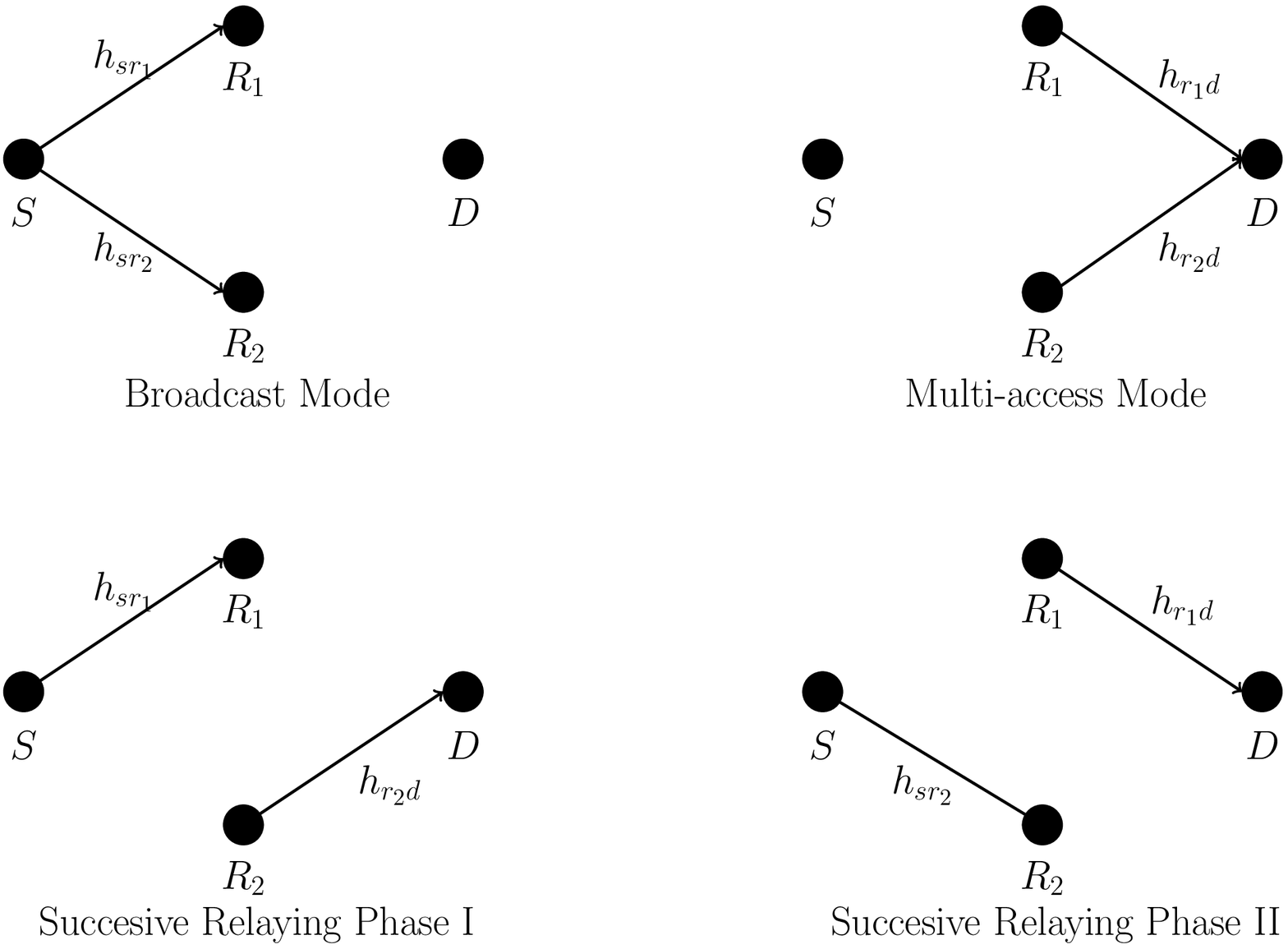}
\label{sys:subfig2}
}
\caption{Transmission modes with (a) one relay, (b) two parallel relays (diamond relay channel).}
\label{sys}%
\end{figure}

\subsection{Achievable Rates}\label{achievable rates}
In this paper, we consider Shannon capacity as the rate-power function of a given link, i.e., $C(p) \triangleq \log(1+\alpha p)$ where $\alpha=|h|^2$ is the power gain of the link and $p$ is the transmission power.

In the single relay case, when the transmission powers of $S$ and $R$ are $p_{s}$ and $p_{r}$, respectively, we have the data rates from $S$ to $R$ (first hop in Figure \ref{sys:subfig1}) and $R$ to $D$ (second hop in Figure \ref{sys:subfig1})  respectively as
\begin{eqnarray}\label{single_s_ach}
c_{s}=\log(1+\alpha_{sr} p_{s}),
\end{eqnarray}
and
\begin{eqnarray}\label{single_r_ach}
c_{r}=\log(1+\alpha_{rd} p_{r}).
\end{eqnarray}

In the two relay case, there are four transmission modes as shown in Figure \ref{sys} \cite{Khandani}. We will assume that $R_1$ and $R_2$ do not beamform towards the destination, hence we will only focus on independent information transmission to and from the relays. While beamforming increases achievable rates, it also requires tighter coordination and synchronization among the relays \cite{Poor}, which may be difficult to achieve for energy harvesting networks typically consisting of small and inexpensive nodes.
\begin{itemize}
\item \textbf{Broadcast mode:} $S$ broadcasts independent information to $R_1$ and $R_2$ resulting in the rate region \cite{infotheory}
\begin{eqnarray}\label{broad ach p}
 c_{b_{r_1}} & \leq & \log(1+\eta \alpha_{sr_1} p_{b})\\\label{broad ach c}
 c_{b_{r_2}} & \leq &  \log\left(1+\frac{(1-\eta)\alpha_{sr_2}p_{b} }{\eta \alpha_{sr_2}p_{b} +1}\right),
\end{eqnarray}
where $p_b$ is the source power used in the broadcast mode and $\eta$ is the power sharing parameter such that $\eta$ portion of the power is used to transmit data to $R_1$. Here, $c_{b_{r_1}}$ is the data rate from $S$ to $R_1$, and $c_{b_{r_2}}$ is the data rate from $S$ to $R_2$.

Operating on the boundary of this rate region, the required transmission power of the source $p_{b}$ can be computed as
\begin{eqnarray}\label{broadcast_power}
p_{b}=\left(\frac{1}{\alpha_{sr_2}}-\frac{1}{\alpha_{sr_1}}\right)e^{c_{b_{r_2}}}-\frac{1}{\alpha_{sr_2}}+\frac{1}{\alpha_{sr_1}}e^{c_{b_{r_1}}+ c_{b_{r_2}}}.
\end{eqnarray}
Note that $p_{b}$ is convex function of $c_{b_{r_1}}$ and $c_{b_{r_2}}$. For notational convenience, we will use $ f_{bc}(c_{b_{r_1}},c_{b_{r_2}})$  to denote the right hand side of (\ref{broadcast_power}).

\item \textbf{Multi-access mode:} $R_1$ and $R_2$ jointly send information to $D$ which uses joint decoding. Denoting the transmission powers of $R_1$ and $R_2$ in multi-access mode as $p_{r_1m}$ and $p_{r_2m}$, respectively, we obtain the following rate region for the multi-access mode \cite{infotheory}
\begin{eqnarray}\label{multi acc 1}
 c_{r_{1}m} &\leq&  \log(1+\alpha_{r_1d} p_{r_{1}m})\\\label{multi acc 2}
 c_{r_{2}m} &\leq&  \log(1+\alpha_{r_2d} p_{r_{2}m})\\\label{multi acc 3}
 c_{r_{1}m}+ c_{r_{2}m} &\leq&  \log(1+\alpha_{r_1d} p_{r_{1}m}+\alpha_{r_2d} p_{r_{2}m})
\end{eqnarray}
where $c_{r_{1}m}$ and $c_{r_{2}m}$ refer to the data rates from $R_1$ and $R_2$ to $D$, respectively.

For notational convenience, we define the following concave and non-decreasing functions.
\begin{eqnarray}\label{multi f 1}
f_{r_{1}m}(p_{r_{1}m})&\triangleq& \log(1+\alpha_{r_1d} p_{r_{1}m}),\\\label{multi f 2}
f_{r_{2}m}(p_{r_{2}m})&\triangleq& \log(1+\alpha_{r_2d} p_{r_{2}m}) ,\\\label{multi f 3}
f_{rm}(p_{r_{1}m},p_{r_{2}m})&\triangleq& \log(1+\alpha_{r_1d} p_{r_{1}m}+\alpha_{r_2d} p_{r_{2}m}).
\end{eqnarray}

\item \textbf{Successive relaying phase I:} While $S$ transmits to $R_1$, $R_2$ transmits to $D$ with transmission powers $p_{sI}$ and $p_{r_2I}$, respectively. Accordingly, the data rates from $S$ to $R_1$ and from $R_2$ to $D$ are given by
    \begin{eqnarray}\label{succ s I}
    c_{sI}=\log(1+\alpha_{sr_1} p_{sI}),
    \end{eqnarray}
    and
    \begin{eqnarray}\label{succ r I}
    c_{r_{2}I}=\log(1+\alpha_{r_2d} p_{r_2I}),
    \end{eqnarray}
    respectively.

\item \textbf{Successive relaying phase II:} While $S$ transmits to $R_2$, $R_1$ transmits to $D$ with transmission powers $p_{sII}$ and $p_{r_2II}$, respectively. Accordingly, the data rates from $S$ to $R_2$ and from $R_1$ to $D$ are given by
    \begin{eqnarray}\label{succ s II}
    c_{sII}=\log(1+\alpha_{sr_2} p_{sII}),
    \end{eqnarray}
    and
    \begin{eqnarray}\label{succ r II}
    c_{r_{1}II}=\log(1+\alpha_{r_1d} p_{r_1II}),
    \end{eqnarray}
    respectively.
\end{itemize}

The following lemmas establish some properties of the optimal transmission policies.

\begin{lemma}\label{lemma 1}
In an epoch, constant power transmission is optimal.
\end{lemma}
\begin{proof}
The proof follows from the concavity of the rate-power functions and Jensen's inequality \cite{infotheory}. First, we argue this for the point-to-point links. Consider any transmission policy for which the transmission power changes in an epoch. We can find another transmission policy which has constant transmission power such that the new policy consumes the same amount of energy as the previous one. However, due to concavity of the rate-power function, the new policy transmits more data \cite{uysal}. For the single relay case ((\ref{single_s_ach})-(\ref{single_r_ach})) and successive relaying phases I and II ((\ref{succ s I})-(\ref{succ r II})) point-to-point rate-power functions apply and hence optimality of constant power transmission in an epoch is established. For the broadcast mode, the proof follows from the strict convexity of the transmission power
as a function of the data rates as given in (\ref{broadcast_power}) and for the multi-access mode, it follows from the concavity of the rate region (\ref{multi acc 1})-(\ref{multi acc 3}) as a function of the transmission powers $p_{r_{1}m}$ and $p_{r_{2}m}$. As a result, we can conclude that the constant transmission policy is optimal for both single and two relay scenarios.
\end{proof}

\begin{lemma}\label{lemma 2}
Given a feasible transmission policy for which a relay is not {\em on}, i.e., not transmitting or receiving data all the time, we can find another feasible transmission policy that ensures the relays are always on without decreasing the throughput.
\end{lemma}
\begin{proof}
Consider a feasible transmission policy for which one of the relays (or the relay in the case of a single relay) is not always on. We can remove the idle times by increasing transmission duration of another node (source or the other relay) while keeping total transmitted data the same. Due to monotonically increasing property of the rate-power functions (\ref{single_s_ach})-(\ref{succ r II}), the new policy delivers the same amount of data to the destination and consumes less energy; hence, it is feasible.
\end{proof}

Lemma \ref{lemma 1} suggests that the transmission powers of the source and the relays remain constant within an epoch. In the following discussion, $i=1,...,K$ refers to the epoch index. Accordingly, for the single relay case we denote the transmission powers of $S$ and $R$ by $p_{s,i}$ and $p_{r,i}$ with corresponding durations $l_{s,i}$ and $l_{r,i}$, respectively. For the case of two relays, we denote the transmission powers of $S$ for the broadcast mode by $p_{b,i}$ with duration $l_{b,i}$. For successive relaying phases I and II the transmission powers of $S$ are denoted by $p_{sI,i}$ and $p_{sII,i}$ with durations $l_{I,i}$ and $l_{II,i}$, respectively. The transmission powers of $R_1$ and $R_2$ in multi-access mode with duration $l_{m,i}$ are denoted by $p_{r_{1}m,i}$ and $p_{r_{2}m,i}$, respectively. For successive relaying phases I and II the transmission powers of $R_1$ and $R_2$ are denoted by $p_{r_1I,i}$ and $p_{r_2II,i}$, respectively. As argued in Lemma \ref{lemma 2}, without loss of generality transmission policies can be restricted to the ones for which the relays $R_1$ and $R_2$ are always on. Therefore, we consider that the transmission time between $S$ and $R_1$, and $R_2$ and $D$ are the same in successive relaying phase I. Similarly, we consider the same transmission time between $S$ and $R_1$, and $R_2$ and $D$ in successive relaying phase II. Accordingly, while evaluating the rates in (\ref{single_s_ach})-(\ref{succ r II}) during an epoch, the corresponding powers in that epoch will be used along with subscripts $i$ in the rate variables  to indicate the epoch index.

%\subsection{Properties of the Optimal Transmission Policy}\label{Properties}
%In this section, we identify some common properties of the optimal transmission policies which will be used to formulate convex optimization problems in Section \ref{two hop} and \ref{two relay} for the one and two relay case, respectively.

%\begin{remark}
%Delaying the transmission of $R$ or $R_1$ ($R_2$) under the half-duplex constraint does not violate energy causality of $R$ or $R_1$ ($R_2$) because postponing the transmission of $R$ or $R_1$ ($R_2$) allows the relays to store more energy. However, this operation is feasible if the energy causality of $S$, and the energy and data causality of $R_2$ ($R_1$) are maintained. Similarly, moving the transmission of $S$ to an earlier time under the half-duplex and source energy causality constraints does not violate the data causality of the $R$ or $R_1$ ($R_2$). However, this is not feasible if the energy and data causality of $R_2$ ($R_1$) are not satisfied.
%\end{remark}

\section{Two-hop Communication with One Relay}\label{two hop}
In this section, we investigate the throughput of the single relay case shown in Figure \ref{system:subfig1}, and consider the achievable rates in (\ref{single_s_ach}) and (\ref{single_r_ach}).
Since constant power transmission in each epoch is optimal by Lemma~\ref{lemma 1}, it is sufficient to consider data causality and buffer size constraints only at energy arrival times.
%
%Note that given any feasible policy, delaying the transmission of $R$ under the half-duplex constraint does not violate feasibility because postponing the transmission of $R$ allows the relay to store more energy and data. Similarly, moving the transmission of $S$ to an earlier time under half-duplex and source energy causality constraints does not violate the data causality. Therefore, feasibility is maintained if transmission of $R$ is switched with a later transmission of $S$. Based on this, we include the data causality constraint only at energy arrival times.

\begin{lemma}\label{lemma 3}
For the single relay case, in an optimal transmission policy, $S$ and $R$ deplete their batteries and transmit same amount of data until the deadline.
\end{lemma}
\begin{proof}
Suppose both $S$ and $R$ have non-zero energy in their batteries at time $T$, and the transmitted data by $S$ is more than $R$. First, we show that given any feasible transmission policy for which the battery of $R$ has nonzero energy at the deadline $T$, we can find another policy which delivers as much as data by depleting all the energy in the battery. This follows from the fact that the rate-power function in (\ref{single_r_ach}) is monotonically increasing function of power. Therefore, the relay can use this excess energy to increase its transmission power $p_{r,K}$ in the last epoch while reducing the transmission duration $l_{r,K}$ such that the delivered data remains the same. While keeping the consumed energy the same, increasing transmission duration strictly increases transmitted data \cite{uysal}. Therefore, the new policy can be replaced by another policy such that the transmission duration of $S$ in the last epoch, $l_{s,K}$, is increased while the last relay transmission is postponed towards to the deadline. Therefore, in the new policy the source delivers more data than the previous policy. We can further increase the total transmitted data by $S$ by depleting all the energy in its battery. As a result, the new policy depletes the batteries of $S$ and $R$ with the source transmitting more data than the initial policy. Now, this policy can be replaced by another one of higher rate obtained by increasing the duration of the last relay transmission while decreasing duration of the preceding source transmission under  data causality and relay buffer size constraints. Combining these, we can find a feasible policy transmitting higher data such that source and relay deplete their batteries and transmit same amount of data until the deadline.
\end{proof}

Based on the above arguments, the throughput optimization problem can be formulated as follows, where the maximization is over  $p_{r,i}$, $p_{s,i}$, $l_{r,i}$, and $l_{s,i}$, $i=1,...,K$:
\begin{subequations}
\label{convex 1}
\begin{align} \label{convex 1:1}
\underset{}{\operatorname{max}}~~ & \sum_{i=1}^{K}{l_{r,i}\log(1+\alpha_{rd}p_{r,i})} \\ \label{convex 1:2}
\text{s.t.}~~ & \sum_{j=1}^{i}{l_{r,j}p_{r,j}} \leq \sum_{j=1}^{i}{E_{r,j}}, ~\forall i, \\ \label{convex 1:3}
& \sum_{j=1}^{i}{{l_{s,j}}p_{s,j}}\leq \sum_{j=1}^{i}{E_{s,j}}, ~ \forall i, \\ \label{convex 1:4}
& \sum_{j=1}^{i}{l_{r,j}\log(1+\alpha_{rd} p_{r,j})} \leq \sum_{j=1}^{i}{{l_{s,j}}\log(1+\alpha_{sr} p_{s,j})}, ~ \forall i,\\ \label{convex 1:45}
& \sum_{j=1}^{i}{{l_{s,j}}\log(1+\alpha_{sr} p_{s,j})} \leq \sum_{j=1}^{i}{l_{r,j}\log(1+\alpha_{rd} p_{r,j})} + B_{max}, ~ \forall i,\\ \label{convex 1:5}
& l_{r,i}+l_{s,i} \leq \tau_{i}, ~  \forall i,\\ \label{convex 1:6}
& 0\leq p_{r,i}, ~ 0\leq p_{s,i}, ~ 0\leq l_{r,i}, ~ 0\leq l_{s,i}, ~ \forall i,
\end{align}
\end{subequations}
where the constraints in (\ref{convex 1:2}), (\ref{convex 1:3}) are due to energy causality at $R$ and $S$, respectively, and the constraints in (\ref{convex 1:4}) and (\ref{convex 1:45}) are due to data causality and finite data buffer size at $R$. The half-duplex constraint appears in (\ref{convex 1:5}). Note that since the total amount of data delivered to $D$ is equal to the amount of data transmitted by $R$, the throughput maximization problem corresponds to maximization of the total data transmitted by $R$ as in (\ref{convex 1:1}) which is equal to (\ref{convex 1:4}) evaluated at $i=K$. The above optimization problem is not convex because of the constraints in (\ref{convex 1:2})-(\ref{convex 1:45}). We rewrite (\ref{convex 1}) in terms of $c_{r,i}$, $c_{s,i}$, $l_{r,i}$, and $l_{s,i}$ as follows:
\begin{subequations}\label{convex 2}
\begin{align}\label{convex 2:1}
\underset{}{\operatorname{max}}~~ & \sum_{i=1}^{K}{c_{r,i}} \\\label{convex 2:2}
\text{s.t.}~~ & \sum_{j=1}^{i}{\frac{l_{r,j}}{\alpha_{rd}}\left(e^{\frac{c_{r,j}}{l_{r,j}}}-1\right)} \leq \sum_{j=1}^{i}{E_{r,j}}, ~ \forall i, \\\label{convex 2:3}
& \sum_{j=1}^{i}{\frac{l_{s,j}}{\alpha_{sr}}\left(e^{\frac{c_{s,j}}{l_{s,j}}}-1\right)} \leq \sum_{j=1}^{i}{E_{s,j}}, ~ \forall i, \\ \label{convex 2:4}
& \sum_{j=1}^{i}{c_{r,j}} \leq \sum_{j=1}^{i}{c_{s,j}}, ~ \forall i,\\ \label{convex 2:45}
& \sum_{j=1}^{i}{c_{s,j}} \leq \sum_{j=1}^{i}{c_{r,j}} + B_{max}, ~ \forall i,\\ \label{convex 2:5}
& l_{r,i}+l_{s,i} \leq \tau_{i}, ~ \forall i,\\ \label{convex 2:6}
& 0\leq c_{r,i}, ~ 0\leq c_{s,i}, ~ 0\leq l_{r,i}, ~ 0\leq l_{s,i}, ~\forall i.
\end{align}
\end{subequations}
Note that $l_{r,i}e^{\frac{c_{r,i}}{l_{r,i}}}$ is perspective of the convex function $e^{c_{r,i}}$, hence it is a convex function of $l_{r,i}$ and $c_{r,i}$ \cite{Boyd}. Here, we consider $l_{r,i}e^{\frac{c_{r,i}}{l_{r,i}}}=0$ when $l_{r,i}=0$. Similarly, ${l_{s,i}}e^{\frac{c_{s,i}}{{l_{s,i}}}}$ in (\ref{convex 2:3}) is a convex function of $l_{s,i}$ and $c_{s,i}$.  Therefore, the optimization problem in (\ref{convex 2}) is convex and can be efficiently solved \cite{Boyd}.

The solution of the optimization problem provides the optimal transmission powers of $S$ and $R$ and their durations for each epoch, but we need to schedule the transmissions to obtain a feasible policy. Within an epoch, moving transmission of source to an earlier time by delaying relay transmission maintains optimality  provided the relay data buffer does not overflow. This is because postponing the transmission of $R$ allows the relay to store more energy and data. Therefore, without loss of optimality, we will consider transmission policies such that in each epoch,  the source transmits until the data buffer of the relay becomes full,  or the source reaches its optimal transmit duration, which is followed by relay transmission until the data buffer of the relay becomes empty, or the relay reaches its optimal transmit duration in that epoch. The source and relay take turns in this fashion until the end of the epoch.

Next, we identify properties of the optimal transmission policy using KKT conditions which are both necessary and sufficient due to convexity of the optimization problem in (\ref{convex 2}). These properties provide the optimal structure of the transmission policy and are useful in designing online algorithms; for example see \cite{elza}.

The Lagrangian of (\ref{convex 2}) is defined as follows:
\begin{align}
\label{lagran 1}
\mathcal{L} =& \sum_{i=1}^{K}{c_{r,i}}-\sum_{i=1}^{K}{\lambda_{1,i} \left(\sum_{j=1}^{i}{\frac{l_{r,j}}{\alpha_{rd}}\left(e^{\frac{c_{r,j}}{l_{r,j}}}-1\right)}-\sum_{j=1}^{i}{E_{r,j}}\right)}\nonumber\\
            & -\sum_{i=1}^{K}{\lambda_{2,i} \left(\sum_{j=1}^{i}{\frac{l_{s,j}}{\alpha_{sr}}\left(e^{\frac{c_{s,j}}{l_{s,j}}}-1\right)}-\sum_{j=1}^{i}{E_{s,j}}\right)} -\sum_{i=1}^{K}{\lambda_{3,i} \left(\sum_{j=1}^{i}{c_{r,j}}-\sum_{j=1}^{i}{c_{s,j}} \right)} \nonumber\\
            &-\sum_{i=1}^{K}{\lambda_{4,i} \left(\sum_{j=1}^{i}{c_{s,j}}-\sum_{j=1}^{i}{c_{r,j}}-B_{max} \right)} -\sum_{i=1}^{K}{\lambda_{5,i} \left(l_{r,i}+l_{s,i}-\tau_i \right)} \nonumber\\
            &  +\sum_{i=1}^{K}{\lambda_{6,i} l_{r,i}} +\sum_{i=1}^{K}{\lambda_{7,i} l_{s,i}}+\sum_{i=1}^{K}{\lambda_{8,i}  c_{r,i}} +\sum_{i=1}^{K}{\lambda_{9,i}  c_{s,i}},
\end{align}
where $\lambda_{j,i}\geq 0$, $j=1, \ldots, 9$ are KKT multipliers corresponding to (\ref{convex 2:2})-(\ref{convex 2:6}).

Differentiating the Lagrangian with respect to $c_{r,i}$ and $c_{s,i}$, we obtain the following:
\vspace{-0.00in}
\begin{align}
\label{der 2}
\frac{\partial \mathcal{L}}{\partial c_{r,i}} &= 1-\frac{e^{\frac{c_{r,i}}{l_{r,i}}}}{\alpha_{rd}}\sum_{j=i}^{K}{\lambda_{1,j}}-\sum_{j=i}^{K}{\lambda_{3,j}}+\sum_{j=i}^{K}{\lambda_{4,j}}+\lambda_{8,i} =0,\\
\label{der 3}
\frac{\partial \mathcal{L}}{\partial c_{s,i}} &= -\frac{e^{\frac{c_{s,i}}{{l_{s,i}}}}}{\alpha_{sr}}\sum_{j=i}^{K}{\lambda_{2,j}}+\sum_{j=i}^{K}{\lambda_{3,j}}-\sum_{j=i}^{K}{\lambda_{4,j}}+\lambda_{9,i} =0.
\end{align}

Using (\ref{der 2}) and replacing $c_{r,i}$ with $l_{r,i}\log(1+\alpha_{rd} p_{r,i})$, we can obtain the optimal relay transmission power $p_{r,i}^*$ as:
\begin{eqnarray}
\label{leq 1}
&p_{r,i}^* = \left[\frac{1-\sum_{j=i}^{K}{\lambda_{3,j}+\sum_{j=i}^{K}\lambda_{4,j}}}{\sum_{j=i}^{K}{\lambda_{1,j}}}-\frac{1}{\alpha_{rd}}\right]^+,
\end{eqnarray}
where $[x]^+=\max\{0,x\}$.

Similarly using (\ref{der 3}) and replacing $c_{s,i}$ with $l_{s,i}\log(1+\alpha_{sr} p_{s,i})$, the optimal source transmission power $p_{s,i}^*$ becomes:
\begin{eqnarray}
\label{leq 2}
&p_{s,i}^* = \left[\frac{\sum_{j=i}^{K}{\lambda_{3,j}-\sum_{j=i}^{K}\lambda_{4,j}}}{\sum_{j=i}^{K}{\lambda_{2,j}}}-\frac{1}{\alpha_{sr}}\right]^+.
\end{eqnarray}

\begin{lemma}\label{lemma 4}
For the single relay case, whenever $p_{r,i}^*$ strictly increases from epoch $i$ to $i+1$, either the battery or the data buffer of $R$ must be empty at $t=t_i$, and whenever $p_{r,i}^*$ strictly decreases from epoch $i$ to $i+1$, the data buffer of $R$ must be full at $t=t_i$.
\end{lemma}
\begin{proof}
We provide a proof using the KKT conditions; alternatively a proof by contradictions as in \cite[Lemmas 4, 5, 7]{oner}, is also possible.  From the complementary slackness conditions, we can argue that whenever $\lambda_{1,i}>0$, the battery of $R$ must be empty at time $t_i$, and whenever $\lambda_{3,i}>0$, the data buffer of $R$ must be empty at time $t_i$. From (\ref{leq 1}), we observe that whenever $p_{r,i}^* < p_{r,i+1}^*$, either $\lambda_{1,i}>0$ or $\lambda_{3,i}>0$ or both, hence proving the lemma. Similarly, from the complementary slackness conditions, we can argue that whenever $\lambda_{4,i}>0$, the data buffer of $R$ must be full at time $t_i$. Since $p_{r,i}^* > p_{r,i+1}^*$ implies $\lambda_{4,i}>0$, the proof is complete.
\end{proof}
\begin{lemma}\label{lemma 5}
For the single relay case the optimal transmission power of $S$ is non-decreasing, and whenever $p_{s,i}^*$ strictly increases from epoch $i$ to $i+1$, either the battery of $S$ must be empty or the data buffer of $R$ must be full, or both the battery of $S$ and the data buffer of $R$ must be empty at $t=t_i$.
\end{lemma}
\begin{proof}
From the complementary slackness conditions, we can argue that  $\lambda_{2,i}>0$ implies the battery of $S$ must be empty at time $t_i$, $\lambda_{3,i}>0$ implies the data buffer of $R$ must be empty at time $t_i$, and  $\lambda_{4,i}>0$ implies the data buffer of $R$ must be full at time $t_i$. Below we investigate different cases for $\lambda_{2,i}$, $\lambda_{3,i}$ and $\lambda_{4,i}$. Since the data buffer of $R$ cannot be full and empty at the same time, the cases ($\lambda_{2,i}=0$, $\lambda_{3,i}>0$, and $\lambda_{4,i}>0$) and ($\lambda_{2,i}>0$, $\lambda_{3,i}>0$, and $\lambda_{4,i}>0$) never happen. Note that $(\lambda_{2,i}=0,\lambda_{3,i}=0)$ and $(\lambda_{2,i}>0,\lambda_{3,i}=0)$ were studied in \cite[Lemma 5]{oner}; a simpler proof using (\ref{leq 2}) is presented here.
\begin{enumerate}
\item If $\lambda_{2,i}=0$, $\lambda_{3,i}=0$, and $\lambda_{4,i}=0$, $p_{s,i}^*=p_{s,i+1}^*$.
\item For the cases ($\lambda_{2,i}>0$, $\lambda_{3,i}=0$, and $\lambda_{4,i}=0$), ($\lambda_{2,i}=0$, $\lambda_{3,i}=0$, and $\lambda_{4,i}>0$), and ($\lambda_{2,i}>0$, $\lambda_{3,i}=0$, and $\lambda_{4,i}>0$), we have $p_{s,i}^*<p_{s,i+1}^*$.
\item For the cases ($\lambda_{2,i}>0$, $\lambda_{3,i}>0$, and $\lambda_{4,i}=0$), and ($\lambda_{2,i}=0$, $\lambda_{3,i}>0$, and $\lambda_{4,i}=0$), we argue by contradiction that $p_{s,i}^*\leq p_{s,i+1}^*$. Note that $\lambda_{2,i}=0$, $\lambda_{3,i}>0$, and $\lambda_{4,i}=0$ implies $p_{s,i}^*>p_{s,i+1}^*$ by (\ref{leq 2}), hence the argument below also suggests that this case never happens.

Suppose $p_{s,i}^*>p_{s,i+1}^*$. We can then equalize the power levels $p_{s,i}^*$ and $p_{s,i+1}^*$ such that the new transmission durations and power levels are $l_{s,i}'=(l_{s,i}+l_{s,i+1})\frac{l_{s,i}p_{s,i}^*}{l_{s,i}p_{s,i}^*+l_{s,i+1}p_{s,i+1}^*}$, $l_{s,i+1}'=l_{s,i}+l_{s,i+1}-l_{s,i}'$, and $p_{s,i}'=p_{s,i+1}'=\frac{p_{s,i}^*+p_{s,i+1}^*}{2}$. The new policy has the same total consumed energy but $S$ transmits more data due to the concavity of the rate-power function. Since we assume that $p_{s,i}^*>p_{s,i+1}^*$, the new transmission duration of $p_{s,i}'$ must increase, i.e., $l_{s,i}'>l_{s,i}$. For the equalized powers, we can obtain another feasible transmission policy by increasing total transmission duration of $R$ and decreasing transmission duration of $S$ and equalizing the transmitted data. As a result, this leads to a policy with higher throughput than the original one, which is a contradiction. Hence, $p_{s,i}^*\leq p_{s,i+1}^*$.
\end{enumerate}
\end{proof}

\section{Two-hop Communication with Two Parallel Relays}\label{two relay}
In this section, we consider the two parallel relay case as shown in Figure \ref{system:subfig2}. We will formulate an optimization problem which includes all four transmission modes given in Section \ref{achievable rates}. Then, to get insights we will investigate special cases by restricting our attention to select few modes.

For ease of exposition, we consider two data buffers, $B_{r_1}$ and $B_{r_2}$ to which data received by $R_1$ and $R_2$ are stored, respectively. The amount of data stored in buffer $B_{r_1}$ in epoch $i$ consists of $c_{b_{r_1},i}$ bits in the broadcast mode and $c_{sI,i}$ bits in successive relaying phase I. The amount of data removed from $B_{r_1}$ in epoch $i$ consists of $c_{r_{1}m,i}$ bits in the multiple access mode and $c_{r_{1}II,i}$ in successive relaying phase II. Similar arguments for buffer $B_{r_2}$ can also be made. Note that $B_{r_1}$ and $B_{r_2}$ are upper bounded by $B_{max}$ bits.

In order to formulate a convex optimization problem for maximizing the throughput we define auxiliary variables $e_{r_{1}m,i}$ and $e_{r_{2}m,i}$, where $e_{r_{1}m,i}=l_{m,i}p_{r_{1}m,i}$ and $e_{r_{2}m,i}=l_{m,i} p_{r_{2}m,i}$. These correspond to the energies allocated by $R_1$ and $R_2$, respectively, to the multiple access phase in epoch $i$. Based on the above arguments, the throughput maximization problem for the two relay case can be formulated as follows:
\begin{subequations}\label{convex 8}
\begin{eqnarray}\label{convex 8:1}
\hspace{-4.3in}\underset{}{\operatorname{max}} && \hspace{-0.2in}  \sum_{i=1}^{K}{c_{r_{1}II,i}+c_{r_{2}I,i}+c_{r_{1}m,i}+c_{r_{2}m,i}}\\\label{convex 8:8}
\text{s.t.} && \hspace{-0.2in}  c_{r_{1}m,i} \leq l_{m,i} f_{r_{1}m}\left(\frac{e_{r_{1}m,i}}{l_{m,i}}\right), \quad \forall i, \\ \label{convex 8:9}
&&\hspace{-0.2in} c_{r_{2}m,i} \leq l_{m,i} f_{r_{2}m}\left(\frac{e_{r_{2}m,i}}{l_{m,i} }\right), \quad \forall i, \\ \label{convex 8:10}
&&\hspace{-0.2in} c_{r_{1}m,i} +c_{r_{2}m,i} \leq l_{m,i} f_{rm}\left(\frac{e_{r_{1}m,i}}{l_{m,i}},\frac{e_{r_{2}m,i}}{l_{m,i} }\right), \quad \forall i, \\ \label{convex 8:2}
&&\hspace{-0.2in} \sum_{j=1}^{i}{l_{b,j}f_{bc}\left(\frac{c_{b_{r_1},j}}{l_{b,j}},\frac{c_{b_{r_2},j}}{l_{b,j}}\right)}+\frac{l_{I,j}}{\alpha_{sr_1}}\left(e^\frac{c_{sI,j}}{l_{I,j}}-1\right)+ \frac{l_{II,j}}{\alpha_{sr_2}}\left(e^\frac{c_{sII,j}}{l_{II,j}}-1\right) \leq \sum_{j=1}^{i}{E_{s,j}},  \forall i, \\\label{convex 8:3}
&&\hspace{-0.2in} \sum_{j=1}^{i}{\frac{l_{II,j}}{\alpha_{r_1d}}\left(e^\frac{c_{r_{1}II,j}}{l_{II,j}}-1\right)} + e_{r_{1}m,j} \leq \sum_{j=1}^{i}{E_{r_1,j}}, \quad \forall i, \\\label{convex 8:4}
&&\hspace{-0.2in} \sum_{j=1}^{i}{\frac{l_{I,j}}{\alpha_{r_2d}}\left(e^\frac{c_{r_{2}I,j}}{l_{I,j}}-1\right)} +e_{r_{2}m,j}\leq \sum_{j=1}^{i}{E_{r_2,j}}, \quad \forall i, \\\label{convex 8:5}
&&\hspace{-0.2in} \sum_{j=1}^{i}c_{r_{1}II,j}+ c_{r_{1}m,j} \leq \sum_{j=1}^{i}{c_{b_{r_1},j}+c_{sI,j}}, \quad \forall i, \\\label{convex 8:6}
&&\hspace{-0.2in} \sum_{j=1}^{i}{c_{r_{2}I,j}}+ c_{r_{2}m,j} \leq \sum_{j=1}^{i}{c_{sII,j}}+c_{b_{r_2},j}, ~ \forall i, \\\label{convex 8:61}
&&\hspace{-0.2in} \sum_{j=1}^{i}{c_{b_{r_1},j}+c_{sI,j}}  \leq  \sum_{j=1}^{i}c_{r_{1}II,j}+ c_{r_{1}m,j}+B_{max}, \quad \forall i, \\\label{convex 8:121}
&&\hspace{-0.2in} \sum_{j=1}^{i}{c_{sII,j}}+c_{b_{r_2},j} \leq  \sum_{j=1}^{i}{c_{r_{2}I,j}}+ c_{r_{2}m,j}+ B_{max}, ~ \forall i, \\\label{convex 8:12}
&&\hspace{-0.2in} l_{I,i}+l_{II,i}+l_{b,i}+l_{m,i} \leq \tau_{i}, \quad  \forall i,\\ \label{convex 8:13}
&&\hspace{-0.2in} 0\leq c_{b_{r_1},i}, ~ 0\leq c_{b_{r_2},i}, ~ 0\leq c_{sI,i}, ~ 0\leq c_{sII,i}, ~ 0\leq c_{r_{1}II,i}, ~ 0\leq c_{r_{2}I,i}, \quad \forall i,\\ \label{convex 8:14}
&&\hspace{-0.2in} 0\leq c_{r_{1}m,i}, ~ 0\leq c_{r_{2}m,i}, ~0\leq l_{b,i}, ~ 0\leq l_{I,i}, ~ 0\leq l_{II,i}, 0\leq l_{m,i}, \quad \forall i,\\\label{convex 8:15}
&&\hspace{-0.2in} 0\leq e_{r_{1_p}m,i},  ~0\leq e_{r_{2_p}m,i}, \quad \forall i,
\end{eqnarray}
\end{subequations}
Here the maximization is over $c_{b_{r_1},i}$, $c_{b_{r_2},i}$, $c_{sI,i}$, $c_{sII,i}$, $c_{r_{1}II,i}$, $c_{r_{2}I,i}$, $c_{r_{1}m,i}$, $c_{r_{2}m,i}$, $l_{b,i}$, $l_{I,i}$, $ l_{II,i}$, $ l_{m,i}$, $e_{r_{1_p}m,i}$, and $e_{r_{2_p}m,i}$. The constraints in (\ref{convex 8:8})-(\ref{convex 8:10}) correspond to the rate region of the multi-access mode as in (\ref{multi acc 1})-(\ref{multi acc 3}). The constraints in (\ref{convex 8:2})-(\ref{convex 8:4}) are the energy causality constraints at $S$, $R_1$, and $R_2$, respectively. The constraints in (\ref{convex 8:5})-(\ref{convex 8:6}) are the data causality constraints at data buffers $B_{r_1}$ and $B_{r_2}$, respectively. The finite data buffer size constraints at $R_1$ and $R_2$ are given in (\ref{convex 8:61})-(\ref{convex 8:121}), respectively. In addition, due to half-duplex constraints, transmission durations $l_{I,i}$, $l_{II,i}$, $l_{b,i}$, and $l_{m,i}$ must satisfy (\ref{convex 8:12}).

As discussed in Section \ref{achievable rates}, $f_{bc}(c_{b_{r_1},i},c_{b_{r_2},i})$ is convex function of $c_{b_{r_1},i}$ and $c_{b_{r_2},i}$. Therefore, $l_{b,i}f_{bc}\left(\frac{c_{b_{r_1},i}}{l_{b,i}},\frac{c_{b_{r_2},i}}{l_{b,i}}\right)$ is the perspective of a convex function. Furthermore, as discussed in Section \ref{two hop}, $le^{\frac{c}{l}}$ is the perspective of the convex function $e^{c}$. In addition, the functions $f_{r_{1}m}$, $f_{r_{2}m}$, $f_{rm}$, and their perspective functions are concave. Hence the optimization problem in (\ref{convex 8}) is convex, and efficient numerical solutions exist \cite{Boyd}. However, due to the large number of variables involved, it is difficult to get insights from the analytical solutions. Below, we focus some special cases: (i) \emph{multi-hop with spatial reuse} in which there are two transmission modes, successive relaying phases I and II; (ii) \emph{broadcast and multi-hop with spatial reuse} in which we have the broadcast mode as well as successive relaying phases I and II. (iii) \emph{multi-access and multi-hop with spatial reuse} in which we have the multi-access mode in addition to the successive relaying phases.

We first focus on multi-hop with special reuse as it is  known to perform well in a wide range of channel conditions and is capacity achieving in certain cases \cite{Sandhu}. Furthermore, it is simple to implement. However, depending on the energy arrival profile and power gains there can be some unused capacity in the first or the second hops \cite{Khandani}. In such cases, we will observe that  adding  the broadcast  or the multiple access modes enables a more efficient use of the harvested energy.

\subsection{Multi-hop with Spatial Reuse}\label{successive relaying}
Multi-hop with spatial reuse refers to successive uses of phase I and II relaying. Our goal in this subsection is to specialize the general formulation of (\ref{convex 8}) to multihop with spatial reuse to identify some of the optimal transmission policy using KKT optimality conditions. Since  $R_2$ initially has no data to transmit in phase I, without loss of generality, we assume it starts transmission by delivering $\epsilon >0$ amount of dummy information. By keeping $\epsilon$ small and scheduling phases I and II in succession, we can ensure that there is no further loss in the throughput. Then, omitting $\epsilon$ for convenience,  the throughput optimization problem can be formulated by setting $l_{b,i}$, $l_{m,i}$, $e_{r_{1}m,i}$,  $e_{r_{2}m,i}$, $c_{r_{1}m,i}$, $c_{r_{2}m,i}$, $c_{b_{r_1},i}$, and $c_{b_{r_2},i}$ in (\ref{convex 8}) to zero for $i=1,...,K$.

As in the single relay case of Section \ref{two hop}, forming the Lagrangian and equating its derivatives to zero we obtain:
\begin{eqnarray}\label{leq b4}
p_{r_{1}II,i}^* &=& \left[\frac{1-\sum_{j=i}^{K}{\lambda_{7,j}-\sum_{j=i}^{K}\lambda_{9,j}}}{\sum_{j=i}^{K}{\lambda_{5,j}}}-\frac{1}{\alpha_{r_1d}}\right]^+,\\\label{leq b3}
p_{r_{2}I,i}^* &=& \left[\frac{1-\sum_{j=i}^{K}{\lambda_{8,j}-\sum_{j=i}^{K}\lambda_{10,j}}}{\sum_{j=i}^{K}{\lambda_{6,j}}}-\frac{1}{\alpha_{r_2d}}\right]^+,\\\label{leq b1}
p_{sI,i}^* &=& \left[\frac{\sum_{j=i}^{K}{\lambda_{7,j}-\sum_{j=i}^{K}\lambda_{9,j}}}{\sum_{j=i}^{K}{\lambda_{4,j}}}-\frac{1}{\alpha_{sr_1}}\right]^+,\\\label{leq b2}
p_{sII,i}^* &= &  \left[ \frac{\sum_{j=i}^{K}{\lambda_{8,j}-\sum_{j=i}^{K}\lambda_{10,j}}}{\sum_{j=i}^{K}{\lambda_{4,j}}}-\frac{1}{\alpha_{sr_2}}\right]^+,
\end{eqnarray}
where $\lambda_{4,i}$, $\lambda_{5,i}$, $\lambda_{6,i}$, $\lambda_{7,i}$, $\lambda_{8,i}$, $\lambda_{9,i}$, and $\lambda_{10,i}$, $i=1,...,K$ are the Lagrange multipliers for the constraints in (\ref{convex 8:2})-(\ref{convex 8:121}), respectively.
\begin{lemma}\label{lemma 6}
For multihop with spatial reuse whenever the optimal transmission power of a relay $p_{r_{1}II,i}^*$ or $p_{r_{2}I,i}^*$ strictly increases, either the battery or the data buffer of that relay must be empty, and whenever the power of  a relay strictly decreases, the data buffer of that relay must be full at time $t=t_i$.
\end{lemma}
\begin{proof}
The proof is similar to that of Lemma~\ref{lemma 4}. We only prove for $R_1$, similar arguments can be made for $R_2$ as well. From (\ref{leq b4}), we have $p_{r_{1}II,i+1}^*-p_{r_{1}II,i}^*>0$, when either $\lambda_{7,i}>0$ or $\lambda_{5,i}> 0$. Using complementary slackness conditions, we know that  $\lambda_{7,i}>0$ implies all the data in $R_1$ must be delivered at the end of the epoch $i$, that is, the data buffer of $R_1$ is empty. Similarly, whenever $\lambda_{5,i}>0$, the battery of $R_1$  must be depleted at the end of the epoch $i$. In addition, we have $p_{r_{1}II,i+1}^*-p_{r_{1}II,i}^*<0$,  when $\lambda_{9,i}>0$. From the complementary slackness conditions, we can argue that whenever $\lambda_{9,i}>0$, the data buffer of $R_1$ must be full at time $t=t_i$. Hence, the lemma must be true.
\end{proof}

\begin{lemma}\label{lemma 8}
For multihop with spatial reuse whenever the optimal transmission power of $S$ in phase I (phase II) strictly increases from one epoch to the next, i.e. $p_{sI,i}^*<p_{sI,i+1}^*$ ($p_{sII,i}^*<p_{sII,i+1}^*$), either the battery of $S$ must be empty, or the data buffer of $R_1$ ($R_2$) must be full at $t=t_i$, and whenever it decreases, i.e. $p_{sI,i}^*>p_{sI,i+1}^*$ ($p_{sII,i}^*>p_{sII,i+1}^*$), the data buffer of $R_1$ ($R_2$) must be empty at $t=t_i$.
\end{lemma}
\begin{proof}
From the complementary slackness conditions, we can argue that whenever $\lambda_{4,i}>0$, the battery of $S$ is empty at $t=t_i$, and whenever $\lambda_{9,i}>0$, the data buffer of $R_1$ must be full at time $t=t_i$.  In addition, whenever $\lambda_{7,i}>0$, the data buffer of $R_1$ is empty at  $t=t_i$. From (\ref{leq b1}), we see that $p_{sI,i}^*<p_{sI,i+1}^*$ implies $\lambda_{4,i}>0$ or $\lambda_{9,i}>0$ and hence the battery of $S$ is empty or the data buffer of $R_1$ is full. Similarly, $p_{sI,i}^*>p_{sI,i+1}^*$ implies $\lambda_{7,i}>0$ and hence the data buffer of $R_1$ is empty. The same argument can be made for phase II and $R_2$ as well.
\end{proof}

Lemma \ref{lemma 6} suggests that the structure of the optimal relay transmission power for the two relay case when multihop with spatial reuse is employed is similar to that of a single relay established in Lemma \ref{lemma 4}. However, comparing Lemma \ref{lemma 8} with Lemma \ref{lemma 5}, we observe that unlike the single relay case where the source power is non-decreasing, in the two relay scenario, the transmission power of the source may decrease when the data buffer of the respective relay is empty.

For the single relay case, as argued in Lemma \ref{lemma 3} batteries of both $S$ and $R$ are depleted by the deadline. This is accomplished by adjusting transmission durations and powers of $S$ and $R$ to equalize the two-hop rates until both batteries are depleted. However, for the case of multi-hop with spatial reuse, simultaneously adjusting the transmission durations of $S$, $R_1$ and $R_2$  to deplete all the batteries may not be possible. Depending on energy profiles at the nodes, the maximum total rate transmitted from $S$ to $R_1$ and $R_2$ can sometimes be more than the total rate $R_1$ and $R_2$ can deliver to $D$, resulting in excess energy at $S$ at $t=T$. Similarly, there may be remaining energy at $R_1$ and/or at $R_2$ at $t=T$. The following lemma discusses this excess energy case.

\begin{lemma}\label{lemma 7}
In the optimal transmission policy for multihop with spatial reuse, if $S$ has positive energy in its battery at $t=T$, then the batteries of both $R_1$ and $R_2$ must be empty.
\end{lemma}
\begin{proof}
The proof is by contradiction. Without loss of generality, assume that in an optimal policy both $S$ and $R_1$ have positive energy in their batteries at $t=T$. Then, we can increase the total data delivered from $S$ to $R_1$ and from $R_1$ to $D$ by increasing the last transmission powers $p_{sI,K}$ and $p_{r_{1}II,K}$, such that all the energies depleted. This results in a contradiction, hence proving the lemma.
\end{proof}

As argued above and in Lemma \ref{lemma 7}, either $S$, or $R_1$ and/or $R_2$ may have positive energy in their batteries at $t=T$. When there is energy left at either of the relays' batteries, the broadcast mode, used in conjunction with multihop with spatial reuse, helps deliver more data to the relay(s), enabling them to use their excess energy. Similarly, when there is excess energy at $S$ at $t=T$, the multi-access mode allows an increase in the data rate the relays can deliver, thus creating an opportunity for S to use its remaining energy.

\subsection{Broadcast and Multi-hop with Spatial Reuse}\label{BC and Successive relaying mode}
In this section, we consider the broadcast mode and successive relaying (phases I and II) jointly. In this case, $S$ can either broadcast to the relays, or can transmit messages at different times using successive relaying. Similar to Section \ref{successive relaying}, we identify properties of the optimal transmission policies using KKT conditions. The throughput maximization problem can be formulated by setting $l_{m,i}$, $e_{r_{1}m,i}$, $e_{r_{2}m,i}$, $c_{r_{1}m,i}$, and $c_{r_{2}m,i}$ in (\ref{convex 8}) to zero for all $i=1,...,K$.

%\begin{lemma} \label{lemma 10}
%In the case of broadcast and multi-hop with spatial reuse, the common information $c_{b_c,i}$, $i=1,...,K$, is delivered only by $R_2$.
%\end{lemma}
%\begin{proof}
%Suppose that $R_1$ delivers part or all of the common information, i.e., $0 < c_{r_{1_c}II,i} \leq c_{b_c,i}$, then we can find another transmission policy in which $c_{b_p,i}' = c_{b_p,i}+c_{r_{1_c}II,i}$ and $c_{b_c,i}'=c_{r_{2_c}I,i}$. The new policy delivers the same amount of data while consuming less energy because $p_{b}$ in (\ref{broadcast_power})  is increasing function of $c_{b_c}$ when $c_{b_c}+c_{b_p}$ is kept constant. Therefore, in the optimal policy $\sum_{i=1}^{K}c_{r_{2_c}I,i}=\sum_{i=1}^{K}c_{b_c,i}$ and $\sum_{i=1}^{K}c_{r_{1_c}II,i}=0$.
%\end{proof}

%Using Lemma \ref{lemma 10}, we can restrict the optimal transmission policy such that the common information is stored in data buffer $B_{r_2}$. Therefore, we can set $c_{r_{1_c}II,i}$ and $c_{r_{2_c}I,i}$ to zero for $i=1,...,K$ and replace the constraint in (\ref{convex 8:6}) by the following.
%\begin{eqnarray}\label{convex 8:s}
%&\sum_{j=1}^{i}{c_{r_{2_p}I,j}}+ c_{r_{2_p}m,j} \leq \sum_{j=1}^{i}{c_{sII,j}+c_{b_c,j}}
%\end{eqnarray}

Formulating the Lagrangian as in Section \ref{successive relaying} with KKT multipliers $\lambda_{4,i}$, $\lambda_{5,i}$, $\lambda_{6,i}$, $\lambda_{7,i}$, $\lambda_{8,i}$, $\lambda_{9,i}$,, and $\lambda_{10,i}$ corresponding to the constraints in (\ref{convex 8:2})-(\ref{convex 8:121}), respectively, we obtain the optimal transmission power of $S$ in the successive relaying modes, $p_{sI,i}^*$ and $p_{sII,i}^*$ as in (\ref{leq b1}) and (\ref{leq b2}), respectively. Similarly, we obtain the optimal transmission powers of $R_1$ and $R_2$ in the successive relaying phase II and I as in (\ref{leq b4}) and (\ref{leq b3}), respectively.

In order to obtain the transmission power of $S$ in the broadcast mode we take the derivative of the Lagrangian with respect to $c_{b_{r_1},i}$ and $c_{b_{r_2},i}$, respectively, and set them to zero.
\begin{eqnarray}
\label{der 12}
\hspace{-0.2in}\frac{\partial \mathcal{L}}{\partial c_{b_{r_1},i}} &\hspace{-0.13in}=&\hspace{-0.15in}-\frac{e^{\frac{c_{b_{r_1},i}+ c_{b_{r_2},i}}{l_{b,i}}}}{\alpha_{sr_1}}\sum_{j=i}^{K}{\lambda_{4,j}}+\sum_{j=i}^{K}{\lambda_{7,j}}-\sum_{j=i}^{K}{\lambda_{9,j}}+\beta_{c_{b_{r_1}},i} =0,\\
\label{der 13}
\hspace{-0.2in}\frac{\partial \mathcal{L}}{\partial c_{b_{r_2},i}} &\hspace{-0.13in}=&\hspace{-0.15in}\left(\left(\frac{1}{\alpha_{sr_1}}-\frac{1}{\alpha_{sr_2}}\right)e^{\frac{c_{b_{r_2},i}}{l_{b,i}}} -\frac{e^{\frac{c_{b_{r_1},i}+ c_{b_{r_2},i}}{l_{b,i}}}}{\alpha_{sr_1}}\right)\sum_{j=i}^{K}{\lambda_{4,j}}+\sum_{j=i}^{K}{\lambda_{8,j}}-\sum_{j=i}^{K}{\lambda_{10,j}}+\beta_{c_{b_{r_2}},i} \hspace{-0.05in}=\hspace{-0.05in}0.
\end{eqnarray}
The KKT multipliers $\beta_{c_{b_{r_1}},i} \geq 0$ and $\beta_{c_{b_{r_2}},i} \geq 0$ are due to non-negativeness of $c_{b_{r_1},i}$ and $c_{b_{r_2},i}$, respectively.

Using (\ref{broadcast_power}), we compute the optimal power of $S$ in the broadcast mode from (\ref{der 12}) and (\ref{der 13}) as
\begin{eqnarray}
\label{der 21}
p_{b,i}^*&=&\frac{\sum_{j=i}^{K}{\lambda_{7,j}-\sum_{j=i}^{K}\lambda_{9,j}}+\beta_{c_{b_{r_1}},i}}{\sum_{j=i}^{K}{\lambda_{4,j}}}-\frac{1}{\alpha_{sr_2}}+\left(\frac{1}{\alpha_{sr_2}}-\frac{1}{\alpha_{sr_1}}\right)e^{\frac{c_{b_{r_2},i}}{l_{b,i}}},\\\label{der 22}
&=&\frac{\sum_{j=i}^{K}{\lambda_{8,j}-\sum_{j=i}^{K}\lambda_{10,j}}+\beta_{c_{b_{r_2}},i}}{\sum_{j=i}^{K}{\lambda_{4,j}}}-\frac{1}{\alpha_{sr_2}}.
\end{eqnarray}

%Next, we consider KKT conditions together with complementary slackness conditions. There are three cases in the broadcast mode:
%\begin{itemize}
%\item If $S$ transmits only to $R_1$, that is, $c_{b_{r_1},i}>0$ and $c_{b_{r_2},i}=0$, then $\beta_{c_{b_{r_1}},i}=0$ in (\ref{der 21}). Therefore, we obtain the optimal power of $S$ as $p_{b,i}^* = p_{sI,i}^*$ where $p_{sI,i}^*$ is equal to (\ref{leq b1}). In this case, without loss of optimality, the broadcast mode can be replaced by the successive relaying phase I with transmission powers $p_{sI,i}^*=p_{b,i}^*$ and $p_{r_2I}^*=0$. Therefore, the properties of $p_{sI,i}^*$ given in Lemma \ref{lemma 8} hold for $p_{b,i}^*$ as well.
%\item If $S$ transmits only to $R_2$, i.e., $c_{b_{r_1},i}=0$ and $c_{b_{r_2},i}>0$, then without loss of optimality we can replace the broadcast mode with the successive relaying phase II with transmission powers $p_{sII,i}^*=p_{b,i}^*$, where $p_{b,i}^*$ is equal to (\ref{leq b2}), and $p_{r_rII}^*=0$.
%\item If $S$ transmits to both $R_1$ and $R_2$, i.e., $c_{b_{r_1},i}>0$ and $c_{b_{r_2},i}>0$, then $\beta_{c_{b_{r_2}},i}=0$ in (\ref{der 22}), and we obtain the optimal power of $S$ is equal to (\ref{leq b2}).
%\end{itemize}
Without loss of generality, we can restrict the optimal transmission policy such that the broadcast mode occurs only when $S$ transmits to both relays. This is because if in the broadcast mode the source only transmits to one of the relays, say $R_1$, then this means $R_2$ will not be {\em on}. Using Lemma \ref{lemma 2}, we can replace this with another strategy for which $R_2$ transmits to the destination while $S$ transmits to $R_1$, thus adding to the duration of successive relaying phase I. Therefore, we have $\eta_i >0$, where $\eta_i$ is the power sharing parameter in the broadcast mode as in (\ref{broadcast_power}).

\begin{lemma}\label{lemma 11}
For broadcast and multihop with spatial reuse, whenever the optimal transmission power of $S$ in broadcast mode strictly increases from one epoch to the next, i.e. $p_{b,i}^*<p_{b,i+1}^*$, either the battery of $S$ must be empty  or the data buffer of $R_2$ must be full at $t=t_i$,  and whenever it decreases, the data buffer of $R_2$ must be empty at $t=t_i$.
\end{lemma}
\begin{proof}
The proof is a simple extension of the proof of  Lemma \ref{lemma 8}.
\end{proof}

\begin{lemma}\label{lemma 9}
In the optimal transmission policy for broadcast and multihop with spatial reuse, whenever the data rate from $S$ to $R_2$ in the broadcast mode increases, i.e., $c_{b_{r_2},i} \leq c_{b_{r_2},i+1}$, where $c_{b_{r_2},i}$ is given in (\ref{broadcast_power}), either the data buffer of $R_1$ or battery of $S$ must be empty, or the data buffer of $R_2$ must be full at $t=t_i$. Whenever the data rate from $S$ to $R_2$ in the broadcast mode decreases, i.e., $c_{b_{r_2},i} \geq c_{b_{r_2},i+1}$, either the data buffer of $R_2$ must be empty or the data buffer of $R_1$ must be full at $t=t_i$.
\end{lemma}
\begin{proof}
Combining (\ref{der 21}) and (\ref{der 22}), we obtain
\begin{eqnarray}
\label{der 23}
\frac{\sum_{j=i}^{K}{\lambda_{8,j}+\sum_{j=i}^{K}\lambda_{9,j}}-\sum_{j=i}^{K}{\lambda_{7,j}-\sum_{j=i}^{K}\lambda_{10,j}}}{\sum_{j=i}^{K}{\lambda_{4,j}}}&=&\left(\frac{1}{\alpha_{sr_2}}-\frac{1}{\alpha_{sr_1}}\right)e^{\frac{c_{b_{r_2},i}}{l_{b,i}}}.
\end{eqnarray}
This follows from the fact that $\beta_{c_{b_{r_1}},i} =0$ and $\beta_{c_{b_{r_2}},i} =0$ when $c_{b_{r_1},i}$ and $c_{b_{r_2},i}$ are positive. From complementary slackness conditions, we can argue that when $\lambda_{4,i}>0$, the battery of $S$ is empty at $t=t_i$, when $\lambda_{8,i}>0$ and $\lambda_{7,i}>0$, the data buffer $B_{r_2}$ and $B_{r_1}$ are empty at $t=t_{i}$, respectively, and when $\lambda_{9,i}>0$ and $\lambda_{10,i}>0$, the data buffer $B_{r_1}$ and $B_{r_2}$ are full at $t=t_{i}$, respectively. Hence  to have $c_{b_{r_2},i}< c_{b_{r_2},i+1}$  either the data buffer of $R_1$ or battery of $S$ must be empty, or the data buffer of $R_2$ must be full. Similarly,  to have $c_{b_{r_2},i}> c_{b_{r_2},i+1}$, either the data buffer of $R_2$ must be empty or the data buffer of $R_1$ must be full.
\end{proof}
%
%Note that as argued in Lemma \ref{lemma 2}, the relays can be always 'on' without loss of optimality. Therefore, for the first and the second case we can find another transmission policy for the successive relaying phase I and II such that $p_{r_2I}^*>0$ and $p_{r_1II}^*>0$, respectively. Therefore, properties obtained in Lemma \ref{lemma 8} for successive relaying phase II holds for the broadcast mode as well.

%Next, we identify a property of optimal transmission policy of $S$ in broadcast mode when $S$ transmits both common and private information.
%\begin{lemma}\label{lemma 9}
%In the optimal transmission policy, whenever the broadcast power decreases, i.e., $p_{b,i}^*>p_{b,i+1}^*$, the data buffer of $R_2$, $B_{r_2}$, must be empty by the end of epoch $i$, and whenever the broadcast power increases, the battery of the source must be empty at $t=t_{i}$.
%\end{lemma}
%\begin{proof}
%From complementary slackness conditions, we can argue that whenever $\lambda_{4,i}>0$, the battery of $S$ is empty at $t=t_i$, and whenever $\lambda_{8,i}>0$, data buffer $B_{r_2}$   is empty at $t=t_{i}$. Therefore, from (\ref{leq b5}), we see that $p_{b,i}^*>p_{b,i+1}^*$ implies $\lambda_{8,i}>0$, hence all the common information is delivered, and $p_{b,i}^*<p_{b,i+1}^*$ implies $\lambda_{4,i}>0$, hence battery of $S$ is empty at $t=t_{i}$.
%\end{proof}

\subsection{Multi-access and Multi-hop with Spatial Reuse}\label{MAC and Successive relaying mode}
In this section, we consider the multi-access mode and successive relaying phases I and II jointly. The throughput maximization problem can be formulated by setting $l_{b,i}$, $c_{b_{r_1},i}$, $c_{b_{r_2},i}$ to zero for $i=1,...,K$ in (\ref{convex 8}).

Formulating the Lagrangian with KKT multipliers $\lambda_{k,i}$, $k=1,...,10$, corresponding to the constraints in (\ref{convex 8:8})-(\ref{convex 8:121}), respectively, for $i=1,...,K$, we obtain the optimal transmission power of $S$ is as in Section \ref{successive relaying}, that is,  $p_{sI,i}^*$ and $p_{sII,i}^*$ are equal to (\ref{leq b1}) and (\ref{leq b2}), respectively. In addition, the optimal transmission powers of $R_1$ and $R_2$ in successive relaying modes, i.e., $p_{r_{1}II,i}^*$ and $p_{r_{2}I,i}^*$, are equal to (\ref{leq b4}) and (\ref{leq b2}), respectively. Accordingly, the properties given in Lemma \ref{lemma 6} and Lemma \ref{lemma 8} also hold in this case.

Next, we obtain the properties of the power allocation in the multi-access mode. Taking derivative of the Lagrangian corresponding to (\ref{convex 8}) with respect to $c_{r_{1}m,i}$, $c_{r_{2}m,i}$, $e_{r_{1}m,i}$, and $e_{r_{2}m,i}$, and setting them to zero we obtain the following.
\begin{eqnarray}
\label{der 14}
\frac{\partial \mathcal{L}}{\partial c_{r_{1}m,i}} &=& 1-\sum_{j=i}^{K}{\lambda_{7,j}}+\sum_{j=i}^{K}{\lambda_{9,j}}-\lambda_{1,i}-\lambda_{3,i} +\beta_{c_{r_{1_p}m},i} = 0,\\
\label{der 15}
\frac{\partial \mathcal{L}}{\partial c_{r_{2}m,i}} &=& 1-\sum_{j=i}^{K}{\lambda_{8,j}}+\sum_{j=i}^{K}{\lambda_{10,j}}-\lambda_{2,i}-\lambda_{3,i} +\beta_{c_{r_{2_p}m},i} = 0,\\
\label{der 18}
\frac{\partial \mathcal{L}}{\partial e_{r_{1}m,i}} &=&-\sum_{j=i}^{K}{\lambda_{5,j}}+
\frac{\lambda_{1,i}l_{m,i}\alpha_{r_1d}}{l_{m,i}+\alpha_{r_1d}e_{r_{1}m,i}}+\frac{\lambda_{3,i}l_{m,i}\alpha_{r_1d}}{l_{m,i}+\alpha_{r_1d}e_{r_{1}m,i}+\alpha_{r_2d}e_{r_{2}m,i}}+\beta_{e_{r_{1}m},i} =0,\\
\label{der 19}
\frac{\partial \mathcal{L}}{\partial e_{r_{2}m,i}} &=&-\sum_{j=i}^{K}{\lambda_{6,j}}+
\frac{\lambda_{2,i}l_{m,i}\alpha_{r_2d}}{l_{m,i}+\alpha_{r_2d}e_{r_{2}m,i}}+\frac{\lambda_{3,i}l_{m,i}\alpha_{r_2d}}{l_{m,i}+\alpha_{r_1d}e_{r_{1}m,i}+\alpha_{r_2d}e_{r_{2}m,i}}+\beta_{e_{r_{2}m},i} =0.
\end{eqnarray}
Here $\beta_{c_{r_{1}m},i}$, $\beta_{c_{r_{2}m},i}$, $\beta_{e_{r_{1}m},i}$, and $\beta_{e_{r_{2}m},i}$ are KKT multipliers due to non-negativeness of $c_{r_{1}m,i}$, $c_{r_{2}m,i}$, $e_{r_{1}m,i}$ and $e_{r_{2}m,i}$, respectively.

Similar to Section IV.B, without loss of generality we can restrict our attention to the cases for which  both $R_1$ and $R_2$ deliver data to $D$ in the multi-access mode. Then $\beta_{c_{r_{1}m},i}$, $\beta_{c_{r_{2}m},i}$, $\beta_{e_{r_{1}m},i}$, and $\beta_{e_{r_{2}m},i}$ in (\ref{der 14})-(\ref{der 19}) are equal to zero. Due to the rate region of multi-access mode defined in constraints in (\ref{convex 8:8})-(\ref{convex 8:10}), the constraint (\ref{convex 8:8}) and/or (\ref{convex 8:9}) can be satisfied with equality, that is, $\lambda_{1,i}>0$ and/or $\lambda_{2,i}>0$.
\begin{itemize}
    \item If $\lambda_{1,i}>0$ and $\lambda_{2,i}=0$, then from (\ref{der 14})-(\ref{der 19}), we obtain
    \begin{eqnarray}\label{eqn 1}
    p_{r_{1}m,i}^*&=&\left[\frac{\sum_{j=i}^{K}{\lambda_{8,j}}+\sum_{j=i}^{K}{\lambda_{9,j}}-\sum_{j=i}^{K}{\lambda_{7,j}}-\sum_{j=i}^{K}{\lambda_{10,j}}}{\sum_{j=i}^{K}{\lambda_{5,j}}-\frac{\alpha_{r_1d}}{\alpha_{r_2d}}\sum_{j=i}^{K}{\lambda_{6,j}}}
    -\frac{1}{\alpha_{r_1d}}\right]^+,
    \end{eqnarray}
     where $\sum_{j=i}^{K}{\lambda_{8,j}}+\sum_{j=i}^{K}{\lambda_{9,j}}>\sum_{j=i}^{K}{\lambda_{7,j}}+\sum_{j=i}^{K}{\lambda_{10,j}}$ and $\sum_{j=i}^{K}{\lambda_{5,j}}>\frac{\alpha_{r_1d}}{\alpha_{r_2d}}\sum_{j=i}^{K}{\lambda_{6,j}}$ since $\lambda_{1,i}>0$ and $\lambda_{2,i}=0$.
     \item  If $\lambda_{1,i}=0$ and $\lambda_{2,i}>0$, then from (\ref{der 14})-(\ref{der 19}), we obtain
     \begin{eqnarray}\label{eqn 3}
    p_{r_{2}m,i}^*&=&\left[\frac{\sum_{j=i}^{K}{\lambda_{7,j}}+\sum_{j=i}^{K}{\lambda_{10,j}}-\sum_{j=i}^{K}{\lambda_{8,j}}-\sum_{j=i}^{K}{\lambda_{9,j}}}{\sum_{j=i}^{K}{\lambda_{6,j}}-\frac{\alpha_{r_2d}}{\alpha_{r_1d}}\sum_{j=i}^{K}{\lambda_{5,j}}}-\frac{1}{\alpha_{r_2d}}\right]^+,
    \end{eqnarray}
     where $\sum_{j=i}^{K}{\lambda_{7,j}}+\sum_{j=i}^{K}{\lambda_{10,j}}>\sum_{j=i}^{K}{\lambda_{8,j}}+\sum_{j=i}^{K}{\lambda_{9,j}}$ and $\sum_{j=i}^{K}{\lambda_{6,j}}>\frac{\alpha_{r_2d}}{\alpha_{r_1d}}\sum_{j=i}^{K}{\lambda_{5,j}}$ since $\lambda_{1,i}=0$ and $\lambda_{2,i}>0$.
     \item Otherwise, we have
   \begin{eqnarray}\label{eqn 4}
    p_{r_{1}m,i}^*&=&\left[\frac{1-\sum_{j=i}^{K}{\lambda_{7,j}}+\sum_{j=i}^{K}{\lambda_{9,j}}}{\sum_{j=i}^{K}{\lambda_{5,j}}}-\frac{\alpha_{r_2d}}{\alpha_{r_1d}}p_{r_{2}m,i}^*-\frac{1}{\alpha_{r_1d}}\right]^+,
    \end{eqnarray}
        \begin{eqnarray}\label{eqn 2}
    p_{r_{2}m,i}^*&=&\left[\frac{1-\sum_{j=i}^{K}{\lambda_{8,j}}+\sum_{j=i}^{K}{\lambda_{10,j}}}{\sum_{j=i}^{K}{\lambda_{6,j}}}-\frac{\alpha_{r_1d}}{\alpha_{r_2d}}p_{r_{1}m,i}^*-\frac{1}{\alpha_{r_2d}}\right]^+.
    \end{eqnarray}
\end{itemize}
Using these, we can identify some properties of the optimal transmission powers of $R_1$ and $R_2$ in the multi-access mode.

\begin{lemma}
In the optimal transmission policy for multi-access and multihop with spatial reuse, the following must be satisfied in the multi-access mode:
\begin{itemize}
\item If the transmission power of $R_1$ ($R_2$) strictly increases from epoch $i$ to $i+1$, i.e., $p_{r_{1}m,i}^*<p_{r_{1}m,i+1}^*$ ($p_{r_{2}m,i}^*<p_{r_{2}m,i+1}^*$), then either the data buffer or the battery of $R_1$ ($R_2$) must be depleted, or the data buffer of $R_2$ ($R_1$) must be full at $t=t_i$.
\item If the transmission powers of both $R_1$ and $R_2$ strictly decrease from epoch $i$ to $i+1$, then  the data buffers of both $R_1$ and $R_2$ must be full at $t=t_i$.
\end{itemize}
\end{lemma}
\begin{proof}
We can argue that whenever $p_{r_{1}m,i}^*<p_{r_{1}m,i+1}^*$, either $\lambda_{7,i}>0$ or $\lambda_{5,i}>0$, or $\lambda_{10,i}>0$ in (\ref{eqn 1}), or either $\lambda_{7,i}>0$ or $\lambda_{5,i}>0$, or $p_{r_{2}m,i}^*>p_{r_{2}m,i+1}^*$ in (\ref{eqn 4}). Similarly we can argue that whenever $p_{r_{2}m,i}^*>p_{r_{2}m,i+1}^*$, either $\lambda_{7,i}>0$, $\lambda_{10,i}>0$ or $\lambda_{5,i}>0$ in (\ref{eqn 3}), or $p_{r_{1}m,i}^*<p_{r_{1}m,i+1}^*$ or $\lambda_{10,i}>0$ in (\ref{eqn 2}). Therefore, we can conclude that if the transmission power of $R_1$ strictly increases from epoch $i$ to $i+1$, then either the data buffer of $R_1$ ($\lambda_{7,i}>0$) or the battery of $R_1$ ($\lambda_{5,i}>0$) must be depleted, or the data buffer of $R_2$ ($\lambda_{7,i}>0$) must be full at the end of epoch $i$. Similarly, the proof can be extended for $R_2$ as well.

Now suppose that the transmission powers of the both $R_1$ and $R_2$ strictly decrease from epoch $i$ to $i+1$, i.e., $p_{r_1m,i}^*>p_{r_1m,i+1}^*$ and $p_{r_{2}m,i}^*>p_{r_{2}m,i+1}^*$. Then, from (\ref{eqn 4}), we observe that $\lambda_{9,i}>0$, and from (\ref{eqn 2}) we see that $\lambda_{10,i}>0$. Therefore, from complementary slackness conditions, we can conclude that the data buffers of both $R_1$ and $R_2$ must be full.
\end{proof}

\section{Illustration of Results}\label{numerical}
In this section, we provide numerical results to show the effect of the number of relays, energy harvesting  and relay buffer size on the optimal throughput. We also compare the performances of various transmission modes in the two relay scenario.

\begin{figure}[t]\label{simm1}
\centering
\includegraphics[scale=0.75,trim= 0 0 0 0]{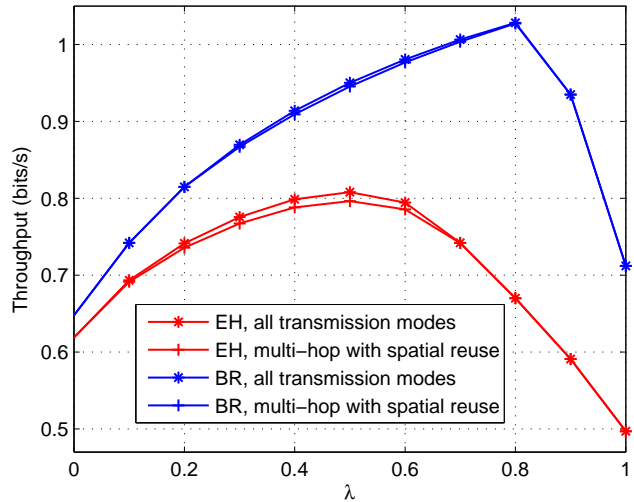}
\caption{Throughput versus $\lambda$, energy allocation among the relays. $R_1$ has total energy $\lambda \mathbf{E_r}$, $R_2$ has $(1-\lambda)\mathbf{E_r}$. $\alpha_{sr_1}=4$, $\alpha_{sr_2}=1$, $\alpha_{r_1d}=1$, and $\alpha_{r_2d}=4$, $T=10$ seconds. For the battery-run (BR) system $E_{s,1}=10$ Joules and $E_r=11.9$ Joules, and for the energy harvesting (EH) system $\mathbf{E_s}=[1,1,1,1,1,1,1,1,1,1]$ Joules and $\mathbf{E_r}=[0.1,0.3,0.3,0.6,0.6,0, 0,1 ,4 ,5]$ Joules with epoch durations $\mathbf{\tau}=[1,0.6,1.4,1.2,0.8,1,1.2,1.6, 0.5,0.7]$ seconds.}
\end{figure}

First, we study the effect of energy harvesting on the throughput of a system with one and with two relays. We consider infinite size data buffer at the relays. We set the power gains to $\alpha_{sr_1}=4$, $\alpha_{sr_2}=1$, $\alpha_{r_1d}=1$, and $\alpha_{r_2d}=4$, and the deadline to $T=10$ seconds. We consider ten epochs with durations $\mathbf{\tau}=[1,0.6,1.4,1.2,0.8,1,1.2,1.6, 0.5,0.7]$ seconds.  We compare the throughputs of the following two scenarios: (i) for each terminal there is a single energy arrival at $t=0$ (\emph{battery-run system}), (ii) for each terminal there are ten energy arrivals at the beginning of the epochs (\emph{energy harvesting system}). For the battery-run system, we have $E_{s,1}=10$ Joules, $E_{r_1,1}=\lambda E_r$, and $E_{r_2,1}=(1-\lambda) E_r $ with $0 \leq \lambda \leq 1$, $E_r=11.9$ Joules, and $E_{s,i}=E_{r_1,i}=E_{r_2,i}=0$, $i=2,...,10$. For the energy harvesting system, source energy arrivals are $\mathbf{E_s}=[1,1,1,1,1,1,1,1,1,1]$ Joules, $R_1$ energies are $\mathbf{E_{r_1}}=\lambda \mathbf{E_r}$, and $R_2$ energies are $\mathbf{E_{r_2}}=(1-\lambda)\mathbf{E_r}$ with $\mathbf{E_r}=[0.1,0.3,0.3,0.6,0.6,0, 0,1 ,4 ,5]$ Joules. Note that in both systems $\lambda = 1$ corresponds to the single relay model with $R_1$ only, and  $\lambda = 0$  with $R_2$ only. Also, the total source and relay energies are same in the battery-run and energy harvesting systems. The throughputs as a function of $\lambda$ for both battery-run and energy harvesting systems are shown in Figure 3. For the case of two relays, we provide the throughputs obtained by optimizing over all four modes, and for multi-hop with spatial reuse only. As expected, the battery-run system with the same total energy performs better than the energy harvesting one. For the channel gains in this particular example, having two relays is always better than having one although this may not be true for arbitrary channel gains due to the energy sharing variable $\lambda$. In addition, for the battery run system,  having only $R_1$ results in more throughput than having only $R_2$, which can be seen by comparing  the throughputs of $\lambda=1$  with $\lambda=0$. This due to the fact that the available energy of $S$ is less than the available energy of the relays; hence, having $\alpha_{sr_1}>\alpha_{r_1d}$ better balances the throughputs in each hop. However, for the energy harvesting system having only $R_1$ results in lower throughput than having only $R_2$. This is because most of the relay energy arrives in the later epochs and hence a higher power gain between the relay and destination is beneficial for the earlier epochs. As shown in the figure, both for the battery run and for the energy harvesting systems, the throughputs when all four modes are considered are slightly higher than the throughputs of multi-hop with spatial reuse  and are equal for large $\lambda$. This is consistent with \cite{Khandani} which shows that multi-hop with spatial reuse obtains most of the capacity gains in many scenarios.

\begin{figure}[t]\label{simm44}
\centering
\includegraphics[scale=0.75,trim= 0 0 0 0]{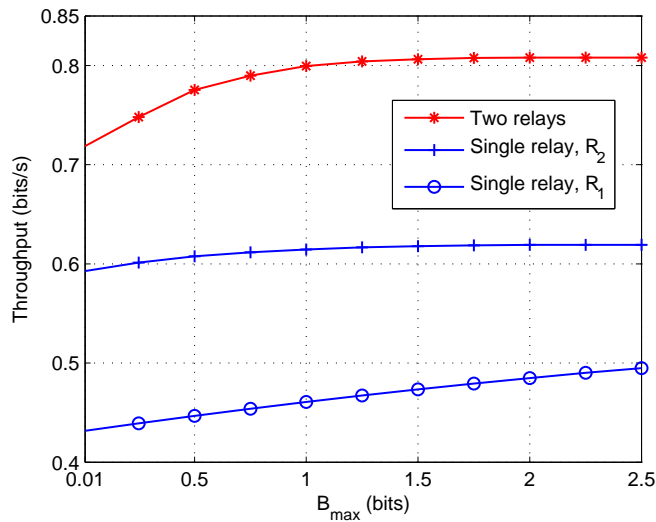}
\caption{ Throughput versus relay data buffer size $B_{max}$. $R_1$ has total energy $\lambda \mathbf{E_r}$, $R_2$ has $(1-\lambda)\mathbf{E_r}$, with $\lambda$ optimized in the two relay case. $\mathbf{E_s}=[1,1,1,1,1,1,1,1,1,1]$ Joules and $\mathbf{E_r}=[0.1,0.3,0.3,0.6,0.6,0, 0,1 ,4 ,5]$ Joules with epoch durations $\mathbf{\tau}=[1,0.6,1.4,1.2,0.8,1,1.2,1.6, 0.5,0.7]$ seconds. $\alpha_{sr_1}=4$, $\alpha_{sr_2}=1$, $\alpha_{r_1d}=1$, and $\alpha_{r_2d}=4$, $T=10$ seconds. }
\end{figure}

We investigate the effect of relay data buffer size on the throughput  in Figure 4.  We consider an energy harvesting system and set the power gains, energy arrivals and epoch durations as  above. We study three cases: (i) Two relays where the throughput is obtained by optimizing over all four transmission modes and energy sharing parameter $\lambda$, (ii)  single relay with  $R_1$ only ($\lambda=1$), (iii)  single relay with $R_2$ only ($\lambda=0$). As shown in the figure, data buffer size is more detrimental for the single relay case with $R_1$ than with $R_2$.  For the case of two relays, for low $B_{max}$, increasing the data buffer capacity leads to a dramatic increase in the throughput. Unlike the scenario with $R_1$ only,  the throughput saturates after $B_{max}>1.75$ bits  when we have two relays since some of the data can be delivered through $R_2$.

\begin{figure}[t]\label{simm2}
\centering
\subfigure[]{
\includegraphics[scale=0.8,trim= 0 0 0 0]{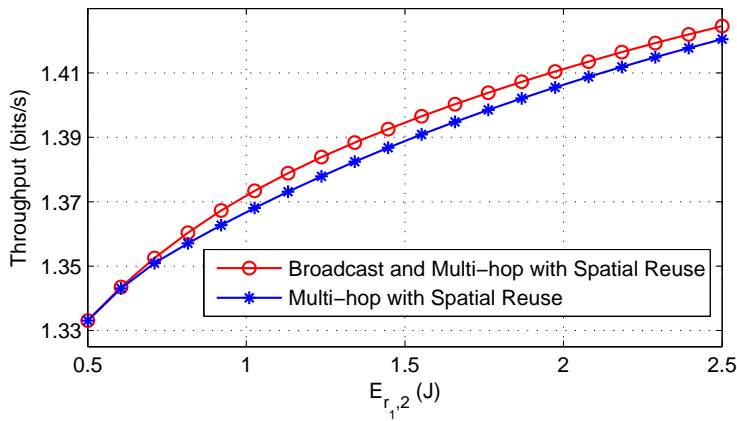}
\label{simm2a}%
}
\subfigure[]{
\includegraphics[scale=0.8,trim= 0 0 0 0]{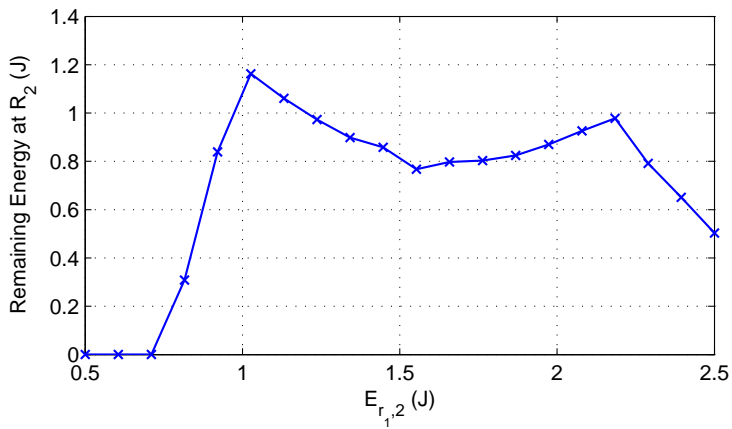}
\label{simm3a}%
}
\caption{Throughput versus relay energy $E_{r_2,2}$. $\mathbf{E_s}=[2.5, 2]$, $\mathbf{E_{r_1}}=[0.5,1.5]$, $\mathbf{E_{r_2}}=[1, E_{r_2,2}]$ Joules where $E_{r_2,2}$ is in the range $(0.5,2.5)$. $\alpha_{sr_1}=2$, $\alpha_{sr_2}=1$, $\alpha_{r_1d}=1$, and $\alpha_{r_2d}=3$, $T=2$ seconds.}
\end{figure}

Next, we compare performances of broadcast and multi-hop with spatial reuse, and multi-hop with spatial reuse only. We set the power gains to $\alpha_{sr_1}=2$, $\alpha_{sr_2}=1$, $\alpha_{r_1d}=1$, and $\alpha_{r_2d}=3$, and the deadline to $T=2$ seconds. We consider an energy harvesting system with two energy arrivals at the beginning of the epochs with durations 1 seconds each. The source energies are $\mathbf{E_s}=[2.5, 2]$ Joules, $R_1$ energies are $\mathbf{E_{r_1}}=[0.5,1.5]$ Joules, and $R_2$ energies are $\mathbf{E_{r_2}}=[1, E_{r_2,2}]$ Joules. Figure 5(a) shows the throughput as a function of $E_{r_2,2}$ which takes values in the range $(0.5,2.5)$ Joules. Figure 5(b) shows the remaining energy at $R_2$ at $T=2$ seconds for multi-hop with spatial reuse. For the above energy and channel profiles the remaining energy at $S$ and $R_1$ are zero. As shown in the figure, when $E_{r_2,2}> 0.72$ Joules, broadcast and multi-hop with spatial reuse performs better than multi-hop with spatial reuse only. This is because for $E_{r_2,2}>0.72$, under multi-hop with spatial reuse protocol, $R_2$ has energy left in its battery at $T=2$ seconds. Introducing the broadcast mode allows the source to send more information to $R_2$, thereby creating an opportunity for $R_2$ to deplete the remaining energy.

\begin{figure}[t]\label{simm3}
\centering
\includegraphics[scale=0.8,trim= 0 0 0 0]{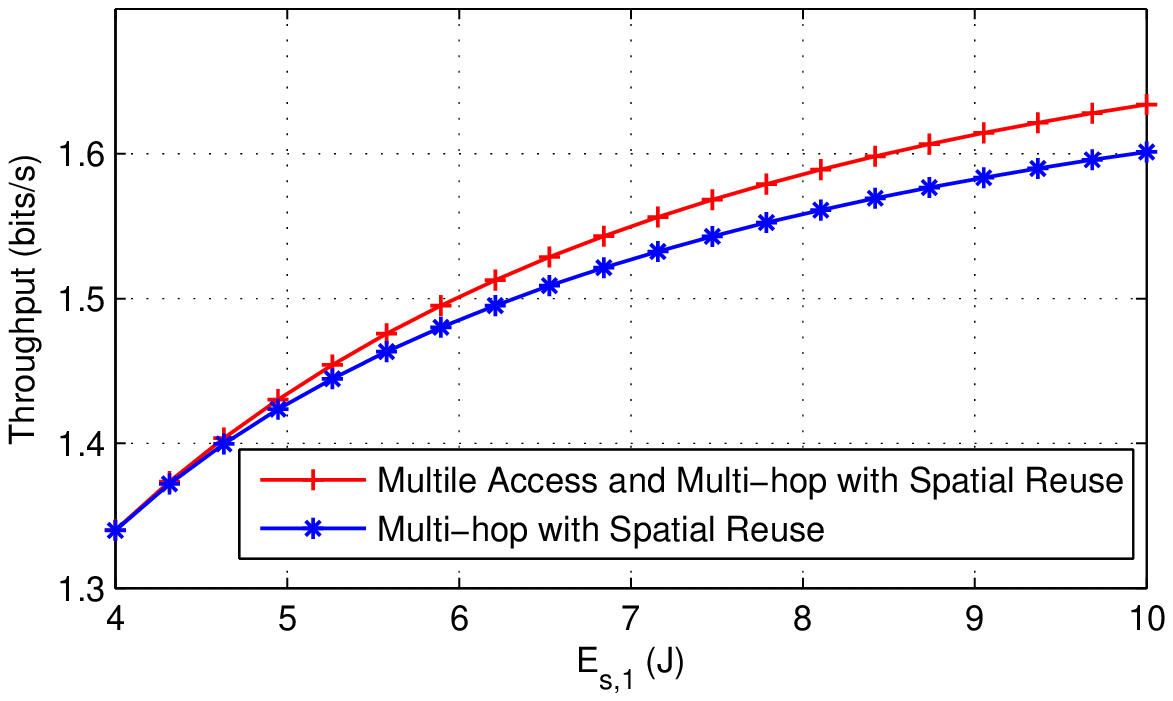}
\caption{Throughput versus source energy $E_{s,1}$. $\mathbf{E_s}=[E_{s,1}, 0]$ where $E_{s,1}$ is in the range $(4,10)$, $\mathbf{E_{r_1}}=[0.01,2]$, $\mathbf{E_{r_2}}=[0.1, 7]$ J. $\alpha_{sr_1}=5$, $\alpha_{sr_2}=1$, $\alpha_{r_1d}=1$, and $\alpha_{r_2d}=3$, $T=2$ seconds.}
\end{figure}

Finally, we compare performances of multi-access and multi-hop with spatial reuse, and multi-hop with spatial reuse schemes. We set the power gains to $\alpha_{sr_1}=5$, $\alpha_{sr_2}=1$, $\alpha_{r_1d}=1$, and $\alpha_{r_2d}=3$, and the deadline to $T=2$ seconds. We consider an energy harvesting system with two energy arrivals at the beginning of epochs of duration 1 second each. The source energies are $\mathbf{E_s}=[E_{s,1}, 0]$ Joules, $R_1$ energies are $\mathbf{E_{r_1}}=[0.01,2]$ Joules, and $R_2$ energies are $\mathbf{E_{r_2}}=[0.1, 7]$ Joules.
In Figure 6, we provide the throughput as a function of $E_{s,1}$ which takes values in the range $(4,10)$ Joules. Note that for the above energy and channel profiles the remaining energy at the nodes are zero for both cases. As shown in the figure, multi-access and multi-hop with spatial reuse performs better than multi-hop with spatial reuse only. This is due to the fact that the multi-access mode makes efficient use of the energy of $R_1$ and $R_2$ to increase the amount of data delivered to the destination. %thereby allowing $S$ to increase the first hop rate using its excess energy.

\section{Conclusion}\label{conclude}
In this paper, we have studied energy harvesting two hop communication with half-duplex relays. We have considered one and two parallel decode-and-forward relays with finite size data buffers employing four transmission modes. Under the assumption of non-causally  known energy arrivals, we have considered optimal transmission policies to maximize the total data delivered by a deadline, and formulated convex optimization problems to compute the throughput. For the case of two relays we have focused on multi-hop with spatial reuse with and without broadcast or multi-access modes. In all cases we have identified various properties of the optimal policies using KKT conditions of the convex optimization formulation. Finally, we have provided performance comparisons and investigated the impact of multiple relays, relay data buffer size, transmission modes and energy harvesting on the average throughput. Overall, our results suggest that while energy harvesting causes a loss in throughput compared with the battery operated scenario, by proper optimization of the transmission power and schedules, it is possible to obtain significant gains. Furthermore, simple relaying strategies such as multihop with spatial reuse are sufficient to obtain a considerable portion of these gains. Possible future extensions include designing online strategies based on the insights gained from the offline solutions provided here and extensions to larger networks involving more relays and more hops.

\end{document}